\documentclass[11pt]{article}

\usepackage[margin=1in]{geometry}
\usepackage{amsmath,amssymb,amsthm}
\usepackage{xcolor}
\usepackage{graphicx} 
\usepackage{float}
\usepackage{placeins}
\usepackage{subfig}
\usepackage{bm}
\usepackage{algorithm}
\usepackage{algpseudocode}
\usepackage[normalem]{ulem}

\newtheorem{Theorem}{Theorem}[section]
\newtheorem{Problem}{Problem}[section]

\newtheorem{Remark}{Remark}[section]

\newtheorem{Corollary}{Corollary}[section]

\title{Forward-Time Black--Scholes Reconstruction via Regularized Legendre Reduction}
\author{
Phuong M. Nguyen\thanks{Department of Mathematics and Statistics, University of North Carolina at Charlotte, Charlotte, NC, 28223, USA} \thanks{Corresponding author. Email: \texttt{pnguye45@charlotte.edu}}
\and
Matt Nguyen\footnotemark[1]
\and
Loc H. Nguyen\footnotemark[1]}

\date{}

\numberwithin{equation}{section}

\begin{document}

\maketitle

\begin{abstract}
We study a forward-time formulation of the Black--Scholes equation with
state-dependent volatility. In contrast to the classical terminal-value
pricing problem, where the option payoff is prescribed at maturity and
the price is computed backward in time, the present problem prescribes
the current option-price profile and seeks to recover the option--price profile
at the expiration date $T$. This formulation is ill-posed, since the equation evolves
in the unstable direction of the parabolic operator and high-frequency
perturbations in the initial data may be strongly amplified. To address
this difficulty, we introduce a price-dimensional reduction based on
shifted Legendre polynomials. The original Black--Scholes equation is
projected onto a finite-dimensional Legendre basis in the asset-price
variable, leading to a system of ordinary differential equations in time
for the expansion coefficients. This reduction acts as a spectral cutoff
and also relaxes the degeneracy caused by the factor $S^2$ at the
zero-price boundary. The main reconstruction method is a
dimension-reduced Legendre--Tikhonov method. We prove existence,
uniqueness, data stability, and convergence for each fixed truncation
level. We also include a reduced PINN solver as a secondary computational
comparison after the Legendre reduction. Numerical experiments with
smooth, butterfly-spread, and European put payoffs show that the
Legendre--Tikhonov method recovers terminal option--price profile from noisy
initial data, while the reduced PINN solver provides a useful additional
benchmark. 
Comparisons with the conventional physical-space
quasi-reversibility method demonstrate the stabilizing effect of the
Legendre reduction.
\end{abstract}



\section{Introduction}

Let $T>0$, and let $u=u(t,S)$ denote the option price, where
$t\in(0,T)$ is time and $S\in(0,\infty)$ is the underlying asset price.
Let $r>0$ be the risk-free interest rate, and let
$\sigma=\sigma(t,S)$ be a state-dependent volatility function. We consider
the Black--Scholes-type equation
\begin{equation}\label{eq:bs-local-vol-forward}
    u_t(t,S)
    + \frac{1}{2}\sigma^2(t,S)S^2 u_{SS}(t,S)
    + rS u_S(t,S)
    - r u(t,S)
    = 0,
    \qquad (t,S)\in (0,T)\times(0,\infty).
\end{equation}

In the classical Black--Scholes framework, the pricing equation is usually
posed as a terminal-value problem. Namely, the option payoff is prescribed
at the expiration time $T$, and the option price is computed for earlier
times. In this paper, we study a different formulation. We prescribe the option-price profile at $t=0$ and seek to compute the corresponding option price profile at maturity $T$ or interchangeably as the terminal option price profile in this paper. More precisely, we consider the
following problem.

\begin{Problem}\label{p}
Given the initial option-price profile
\[
    u(0,S)=u^0(S), \qquad S\in(0,\infty),
\]
determine the option-price profile at maturity $T$
\[
    \Phi(S)=u(T,S), \qquad S\in(0,\infty).
\]
\end{Problem}

Problem~\ref{p} is fundamentally different from the standard
Black--Scholes pricing problem. In the usual formulation, one prescribes
the terminal payoff $u(T,S)=\Phi(S)$ and computes $u(t,S)$ for $t<T$.
Although this is often described as solving the Black--Scholes equation
backward in calendar time, it becomes a well-posed forward parabolic
problem after introducing the time-to-maturity variable $\tau=T-t$. In
contrast, Problem~\ref{p} evolves the Black--Scholes equation in the
opposite, unstable parabolic direction. Therefore, small perturbations or
noise in the initial profile $u^0$ may be strongly amplified at later
times. This makes Problem~\ref{p} an ill-posed finite-time prediction
problem, and regularization is needed.

The Black--Scholes model is one of the foundational frameworks of modern
financial mathematics \cite{black1973pricing}. Its classical form assumes
constant volatility and gives a closed-form formula for European options.
However, market data often exhibit implied-volatility patterns that cannot
be captured by a constant volatility model. This motivated the local
volatility framework introduced in \cite{dupire1994pricing}, in
which the volatility depends on both time and the asset price. Such models
lead to Black--Scholes-type partial differential equations with a
state-dependent volatility function $\sigma(t,S)$.

Many numerical methods have been developed for option-pricing problems
based on Black-Scholes-type equations. Finite difference methods have
been widely used for computing option-price profiles
\cite{cen2011robust,jeong2013comparison,jeong2018finite,kim2020finite,lee2023accurate},
and finite element methods have also been studied
\cite{andalaft2011real,golbabai2013superconvergence, suliadaptive, kumar2025efficient}.
For high-dimensional or path-dependent derivatives, Monte Carlo methods
are often preferred because of their flexibility and scalability
\cite{boyle1977options,chen2023quantum,glasserman2004monte,jeong2019hybrid,longstaff2001valuing}.
More recently, physics-informed neural networks (PINNs), introduced in
\cite{raissi2019physics}, have been applied to option-pricing problems
\cite{kim2024physics,wang2023deep}. Additionally,
\cite{dhiman2023physics} uses PINNs for European call and American put
options, while \cite{bae2024option} studies PINN approximations of
European put option prices and option
sensitivities such as Delta and Gamma.

Forward-time option-price prediction has also been studied in connection
with ill-posed parabolic problems. The work \cite{klibanov2015profitable}
formulates option-price forecasting as an ill-posed forward-in-time
problem for the Black--Scholes equation and solves it by regularization.
The later work \cite{Klibanov_2022} develops a quasi-reversibility method
for the positive-time Black--Scholes equation and proves convergence using
a Carleman estimate. Related machine-learning approaches have also been
considered; for instance, neural networks and convolutional neural
networks were used in \cite{KlibanovKirill2022,cao2023application} for
option-price forecasting where they showed improved forecasting performance compared with the quasi-reversibility method. More broadly, quasi-reversibility methods are
classical regularization tools for ill-posed problems
\cite{bourgeois2010quasi,darde2015iterated,klibanov1991computational,li2020recovering,nguyen2019inverse}.
Inverse problems associated with Black--Scholes-type equations, especially
volatility reconstruction and parameter calibration, have also been
studied by Tikhonov-type methods and related regularization techniques
\cite{bouchouev1997inverse,bouchouev1999uniqueness,crepey2003calibration,egger2005tikhonov,hein2005some,kwak2022reconstructing}.

The main contribution of this paper is a price-dimensional regularized
reconstruction framework for Problem~\ref{p}. The contribution is mainly
computational and methodological: we aim to obtain stable future-time
reconstructions for an ill-posed Black--Scholes problem whose diffusion
coefficient degenerates at $S=0$. We first truncate the asset-price domain
to a finite interval and project the solution onto a shifted Legendre basis
in the price variable $S$. Retaining only the first $N+1$ modes converts
the original Black--Scholes partial differential equation into a
finite-dimensional system of ordinary differential equations in time for
the Legendre expansion coefficients. This reduction acts as a spectral
cutoff: it suppresses highly oscillatory components that are responsible
for the strongest instability of the forward-time problem. It also
alleviates the difficulty caused by the degeneracy of the Black--Scholes
coefficient $\frac12\sigma^2(t,S)S^2$ at $S=0$, since the principal
differential operator in the reduced system is the first derivative with
respect to time, while the effect of the degenerate spatial operator is
encoded in the coefficient matrix.

After deriving the reduced coefficient system, we focus on a
dimension-reduced Legendre--Tikhonov method. In this method, the
coefficient vector is obtained by minimizing a Tikhonov functional
containing the residual of the reduced ordinary differential system, the
projected initial data, and an $H^2$ regularization term. For this method,
we prove existence and uniqueness of the minimizer, stability with respect
to the projected initial data, and convergence as the noise level tends to
zero for each fixed truncation number $N$. Thus, the rigorous
regularization analysis in this paper is carried out for the
Legendre--Tikhonov reconstruction.

We also include a reduced PINN solver as a secondary computational
approach. The PINN is applied only after the Legendre reduction, so the
network depends on the scalar time variable and outputs the first $N+1$
Legendre coefficients. This neural-network component is not used in the
convergence analysis. It is included to compare a neural-network
parametrization with the theoretically justified Legendre--Tikhonov
reconstruction and to test whether such a parametrization can provide a
useful reduced solver in practice.

Dimensional reduction has also been used as an effective tool in inverse
and ill-posed problems. The main idea is to project the unknown function
or the governing equation onto a finite-dimensional space, thereby
reducing a high-dimensional problem to a more manageable system while
retaining the dominant features of the solution. A related direction was
developed in \cite{dang2025recovery}, where a Legendre-polynomial
exponential basis was introduced for inverse initial data problems. This
type of basis has been successfully applied to several inverse problems
for partial differential equations
\cite{le2026globally,le2025inverse,neupane2026inverse,nguyen2026carleman,van2026inverse}.
Motivated by these works, we use shifted Legendre polynomials in the
asset-price variable to reduce the forward-time Black--Scholes equation
to a finite-dimensional system of ordinary differential equations in time.

We validate the proposed framework through several numerical experiments.
The tests include a smooth compactly supported profile, a European
butterfly spread payoff, and a European put payoff. These examples cover
both smooth and nonsmooth profiles, compactly supported and non-compact
payoff structures, and financially standard option payoffs. The numerical
experiments are designed to test the method in challenging settings,
including moderate final times such as $T=1$, $T=1.5$, and $T = 3$, and substantial
noise in the observed initial profile, with one test involving $35\%$
noise.

The remainder of the paper is organized as follows. In
Section~\ref{sec2}, we highlight the main challenges of the problem,
including the forward-time ill-posedness and the degeneracy at the zero
asset-price boundary. In Section~\ref{sec3}, we introduce the
price-dimensional reduction based on shifted Legendre polynomials and
derive the reduced ordinary differential system. In Section~\ref{sec4},
we present the Legendre--Tikhonov method, prove its stability and
convergence, and describe a reduced PINN solver used as a secondary
computational comparison. In Section~\ref{sec5}, we describe the
numerical setup, data generation, parameter selection, numerical examples and show comparisons with the conventional physical-space quasi-reversibility method. Finally, Section~\ref{sec6} contains concluding remarks.

\section{The main challenges}\label{sec2}

The initial-value formulation in Problem \ref{p} involves two principal
difficulties. The first is the severe instability caused by evolving the
Black--Scholes equation in the unfavorable parabolic direction. The second
is the degeneracy of the Black--Scholes operator at the boundary
$S=0$.

\subsection{Ill-posedness of the forward-time formulation}\label{ill-posedness}

The main difficulty of Problem \ref{p} is its ill-posedness. Although
\eqref{eq:bs-local-vol-forward} has the form of the Black--Scholes pricing
equation, prescribing data at $t=0$ and computing the solution at a later
time $T$ reverses the natural parabolic direction of the equation. Indeed,
rewriting \eqref{eq:bs-local-vol-forward} gives
\[
    u_t
    =
    -\frac{1}{2}\sigma^2(t,S)S^2u_{SS}
    -rS u_S
    +ru.
\]
Thus the second-order term has the opposite sign from a forward parabolic
evolution. In other words, the initial-value formulation behaves like a
backward heat equation in the asset-price variable. As a consequence,
high-frequency components of the initial data may be strongly amplified
as time increases, and small perturbations in $u_0$ can lead to large
changes in the computed profile $u(T,\cdot)$. Therefore, the mapping
$u_0\mapsto u(T,\cdot)$ is not expected to be continuous in standard
norms, and the problem is ill-posed in the sense of Hadamard.

This instability can be seen clearly in the constant-volatility case.
Assume that $\sigma>0$ is constant and introduce the logarithmic change of
variables $x=\log S$ and $w(t,x)=u(t,e^x)$. Then the Black--Scholes
equation becomes
\[
    w_t
    =
    -\frac{1}{2}\sigma^2 w_{xx}
    -\left(r-\frac{1}{2}\sigma^2\right)w_x
    +rw.
\]
The important point is the negative sign in front of $w_{xx}$. This is
opposite to the usual heat equation, where diffusion smooths out
oscillations. Here, because the sign is reversed, oscillations are
amplified rather than damped.

To make this instability precise, consider a small oscillatory perturbation
in the initial data. Let $w$ be the exact solution and let $w^\varepsilon$
be the solution corresponding to the perturbed initial data
$w^\varepsilon(0,x)=w(0,x)+\varepsilon e^{i\xi x}$, where
$0<\varepsilon\ll1$ and $\xi\in\mathbb{R}$ is the spatial frequency.
Denote the error by $z(t,x)=w^\varepsilon(t,x)-w(t,x)$. Since the equation
is linear, $z$ satisfies
\[
    z_t
    =
    -\frac{1}{2}\sigma^2 z_{xx}
    -\left(r-\frac{1}{2}\sigma^2\right)z_x
    +rz,
    \qquad
    z(0,x)=\varepsilon e^{i\xi x}.
\]
Solving this constant-coefficient equation for the Fourier component
directly, we obtain
\[
    z(t,x)
    =
    \varepsilon
    \exp\left(
        \left[
            \frac{1}{2}\sigma^2\xi^2
            -i\left(r-\frac{1}{2}\sigma^2\right)\xi
            +r
        \right]t
    \right)e^{i\xi x}.
\]
Hence
\[
    |z(t,x)|
    =
    \varepsilon
    \exp\left(
        \left[
            \frac{1}{2}\sigma^2\xi^2+r
        \right]t
    \right).
\]
Thus a perturbation of size $\varepsilon$ at frequency $\xi$ is amplified
by a factor containing $\exp(\frac{1}{2}\sigma^2\xi^2t)$. As
$|\xi|\to\infty$, this factor grows extremely fast. Therefore,
arbitrarily small high-frequency noise in the initial profile can become
large at later times. This demonstrates the instability of the
forward-in-time initial-value formulation.

The above example shows that the ill-posedness is mainly caused by the
amplification of highly oscillatory components of the noise. Indeed, a
small perturbation with frequency $\xi$ is amplified by a factor containing
$\exp(\frac12\sigma^2\xi^2t)$, which becomes very large as $|\xi|$
increases. Therefore, the most unstable part of the data is its
high-frequency component. In this work, we address this difficulty by
projecting the solution in the price variable onto a finite-dimensional
space and retaining only finitely many modes. The Black--Scholes equation
is then approximated by a system of ordinary differential equations in
time for the expansion coefficients. This dimension reduction plays the
role of a spectral cutoff: it removes the highly oscillatory modes that
are responsible for the strongest instability while preserving the
dominant low-frequency structure of the solution.

\subsection{Degeneracy at the zero asset-price boundary}

Another difficulty comes from the degeneracy of the Black--Scholes
operator at the boundary $S=0$. The leading second-order coefficient in
\eqref{eq:bs-local-vol-forward} is
$
    \frac{1}{2}\sigma^2(t,S)S^2.
$
Even when the volatility is strictly positive, this coefficient vanishes
as $S\to0^+$. Hence the equation is not uniformly parabolic on a domain
touching $S=0$. This lack of uniform parabolicity creates additional
analytical and numerical difficulties. Many standard tools for parabolic
equations, including energy estimates, stability estimates, regularity
arguments, and Carleman estimates, are usually formulated under the
assumption that the leading coefficient is bounded away from zero. When
this condition fails, the behavior near the degenerate boundary must be
handled separately, either by using estimates adapted to degenerate
operators, by transforming the equation, or by restricting the analysis to
an interval away from $S=0$.

The works \cite{KlibanovKirill2022,Klibanov_2022} provide important
quasi-reversibility and Carleman-estimate approaches for positive-time
Black--Scholes problems. In those papers, the stock-price variable is
available on an interval
\[
    S\in(s_b,s_a), \qquad 0<s_b<s_a.
\]
On such an interval, the factor $S^2$ is bounded away from zero. Hence,
after an affine change of variables to $(0,1)$, the leading coefficient
in the transformed equation satisfies a uniform positivity condition.

The formulation considered here includes the degenerate endpoint
$S=0$. Thus the coefficient
\[
    \frac12 \sigma^2(t,S)S^2
\]
vanishes at the boundary, and the equation is not uniformly parabolic on
a domain touching $S=0$. This prevents the direct use of standard
Carleman estimates for uniformly parabolic equations. To handle this
difficulty, we apply a Legendre reduction in the asset-price variable.
The resulting reduced problem is a finite-dimensional system of ordinary
differential equations in time. In this system, the principal
differential operator is the first derivative with respect to time, while
the effect of the degenerate Black--Scholes spatial operator is encoded
in the coefficient matrices.

\section{Legendre dimension reduction}\label{sec3}

Although the model is naturally stated for $S\in(0,\infty)$, the
computations are carried out on a truncated asset-price interval
$(0,S_{\rm max})$. It is therefore natural to use basis functions that are
orthogonal on this finite interval. Let $P_n$ be the $n$th Legendre
polynomial on $(-1,1)$, defined by Rodrigues' formula
\[
    P_n(x)
    =
    \frac{1}{2^n n!}
    \frac{d^n}{dx^n}
    \left(x^2-1\right)^n,
    \qquad -1<x<1.
\]
The polynomials $\{P_n\}_{n=0}^\infty$ are orthogonal in $L^2(-1,1)$ and
satisfy
\[
    \int_{-1}^1 P_n(x)P_m(x)\,dx
    =
    \frac{2}{2n+1}\delta_{mn},
\]
where $\delta_{mn}$ is the Kronecker delta. We define the shifted and
normalized Legendre functions on $(0,S_{\rm max})$ by
\begin{equation}\label{eq:Legendre-functions}
    \ell_n(S)
    =
    \sqrt{\frac{2n+1}{S_{\rm max}}}\,
    P_n\left(\frac{2S}{S_{\rm max}}-1\right),
    \qquad 0<S<S_{\rm max}, \quad n=0,1,2,\dots .
\end{equation}
Then
\[
    \int_0^{S_{\rm max}}\ell_n(S)\ell_m(S)\,dS=\delta_{mn},
    \qquad m,n=0,1,2,\dots,
\]
and $\{\ell_n\}_{n=0}^\infty$ forms a complete orthonormal system in
$L^2(0,S_{\rm max})$.

The Fourier expansion of $u(t, S)$ with respect to the basis $\{\ell_n\}_{n \geq 0}$ is
\begin{equation}\label{3.2}
	u(t, S) = \sum_{n = 0}^\infty u_n(t)\ell_n(S)
	\qquad
	\mbox{where }
	\qquad u_n(t) = \int_0^{S_{\rm max}}u(t,S)\ell_n(S)\,dS 
\end{equation}

The following result is a standard consequence of the convergence of
Legendre projections in Sobolev spaces; see, for example,
\cite{CanutoQuarteroni82,ShenTangWang11}. The only minor differences here
are that we use shifted and normalized Legendre functions on
$(0,S_{\rm max})$ and allow the expansion coefficients to depend on time.
These changes follow directly from an affine change of variables and from
applying the projection result to $H^1(0,T)$-valued functions. For
convenience, we state the result in the notation of this paper.

\begin{Theorem}\label{thm31}
Assume that
\[
    u\in H^k((0,S_{\rm max});H^1(0,T)),
    \qquad k\geq 3.
\]
Let $u_n(t)$ be as in \eqref{3.2}.
Then
\begin{equation}\label{3.3}
    u_S(t,S)=\sum_{n=0}^\infty u_n(t)\ell_n'(S),
    \qquad
    u_{SS}(t,S)=\sum_{n=0}^\infty u_n(t)\ell_n''(S),
\end{equation}
with convergence in $L^2((0,S_{\rm max});H^1(0,T))$. In particular, the
identities hold almost everywhere in $(0,T)\times(0,S_{\rm max})$.
\end{Theorem}

Substituting \eqref{3.2} and \eqref{3.3} into \eqref{eq:bs-local-vol-forward} gives
\begin{equation}\label{eq:Legendre-residual}
    \sum_{n=0}^\infty u_n'(t)\ell_n(S)
    +\frac12 \sigma^2(t,S)S^2
    \sum_{n=0}^\infty u_n(t)\ell_n''(S)
    +rS\sum_{n=0}^\infty u_n(t)\ell_n'(S)
    -r\sum_{n=0}^\infty u_n(t)\ell_n(S)
    =0
\end{equation}
for $(t, S) \in (0, T) \times (0, S_{\rm max}).$
Fix $m\in \mathbb{N}$. Multiplying \eqref{eq:Legendre-residual} by
$\ell_m(S)$ and integrating over $(0,S_{\rm max})$ gives
\begin{multline}\label{eq:projected-system-weak}
\sum_{n=0}^\infty u_n'(t)\int_0^{S_{\rm max}} \ell_n(S)\ell_m(S)\,dS
+\frac12\sum_{n=0}^\infty u_n(t)
\int_0^{S_{\rm max}} \sigma^2(t,S)S^2\ell_n''(S)\ell_m(S)\,dS  \\
+r\sum_{n=0}^\infty u_n(t)
\int_0^{S_{\rm max}} S\ell_n'(S)\ell_m(S)\,dS
-r\sum_{n=0}^\infty u_n(t)
\int_0^{S_{\rm max}} \ell_n(S)\ell_m(S)\,dS
=0,
\end{multline}
for all $t \in (0, T).$
Since $\{\ell_n\}_{n=0}^\infty$ is orthonormal in $L^2(0,S_{\rm max})$, \eqref{eq:projected-system-weak} becomes
\begin{equation}\label{eq:projected-system}
    u_m'(t)
    +\frac12\sum_{n=0}^\infty A_{mn}(t)u_n(t)
    +r\sum_{n=0}^\infty B_{mn}u_n(t)
    -r u_m(t)
    =0,
    \qquad m \in \mathbb{N},
\end{equation}
where
\[
    A_{mn}(t)=
    \int_0^{S_{\rm max}} \sigma^2(t,S)S^2\ell_n''(S)\ell_m(S)\,dS,
    \qquad
    B_{mn}=
    \int_0^{S_{\rm max}} S\ell_n'(S)\ell_m(S)\,dS .
\]
Equivalently,
\begin{equation}\label{eq:projected-ode}
    u_m'(t)
    =
    -\frac12\sum_{n=0}^\infty A_{mn}(t)u_n(t)
    -r\sum_{n=0}^\infty B_{mn}u_n(t)
    +r u_m(t), 
    \quad m \in \mathbb{N}.
\end{equation}
In computation, we fix a cutoff number $N$ and approximate \eqref{eq:projected-ode} as
\begin{equation}\label{eq:projected-odeN}
    u_m'(t)
    =
    \sum_{n = 0}^N C_{mn}(t)  u_n(t)
    \quad m \in \{0, 1, \dots, N\}.
\end{equation}
where
\[
    C_{mn}(t)=-\frac12 A_{mn}(t)-rB_{mn}+r\delta_{mn}
\]
and $\delta_{mn}$ is the Kronecker delta. 
Equivalently, if
\[
    \mathbf u(t)=\big(u_0(t),u_1(t),\dots,u_N(t)\big)^\top,
\]
the finite-dimensional system takes the compact form
\begin{equation}\label{eq:projected-ode-vector}
    \mathbf u'(t)=C(t)\mathbf u(t).
\end{equation}
Together with the projected initial data
\begin{equation}\label{eq:projected-initial-data}
    u_m(0)=\int_0^{S_{\rm max}} u^0(S)\ell_m(S)\,dS,
    \qquad m=0,\dots,N,
\end{equation}
which gives the initial-value problem for the coefficient vector $\mathbf u(t)$.

\begin{Remark}[Consistency of the finite Legendre reduction]
The finite system \eqref{eq:projected-ode-vector} should be understood as
a truncation of the infinite projected Legendre system. Let $P_N$ denote
the $L^2(0,S_{\rm max})$-orthogonal projection onto
$\operatorname{span}\{\ell_0,\ldots,\ell_N\}$. For the exact solution
$u$, the first $N+1$ Legendre coefficients satisfy
\[
    {\bf u}_N'(t)=C(t){\bf u}_N(t)+{\bf E}_N(t),
    \qquad
    {\bf u}_N(t)=(u_0(t),\ldots,u_N(t))^\top ,
\]
where the tail term is
\[
    ({\bf E}_N(t))_m
    =
    \int_0^{S_{\rm max}}
    \mathcal L(t)(I-P_N)u(t,S)\ell_m(S)\,dS,
    \qquad m=0,\ldots,N,
\]
and
\[
    \mathcal L(t)w
    =
    -\frac12\sigma^2(t,S)S^2w_{SS}
    -rS w_S
    +rw .
\]
Thus \eqref{eq:projected-ode-vector} is obtained by neglecting the
Legendre tail ${\bf E}_N$. If $\mathcal L(t)$ is uniformly bounded from
$H^2(0,S_{\rm max})$ to $L^2(0,S_{\rm max})$, then
\[
    \|{\bf E}_N\|_{L^2(0,T;\mathbb R^{N+1})}
    \leq
    C
    \|(I-P_N)u\|_{L^2(0,T;H^2(0,S_{\rm max}))}.
\]
By Theorem~\ref{thm31}, the Legendre expansion and its first two
$S$-derivatives converge under the stated regularity assumptions. Hence
\[
    \|{\bf E}_N\|_{L^2(0,T;\mathbb R^{N+1})}\to0
    \quad
    \text{as } N\to\infty .
\]
Therefore, the finite reduced system is consistent with the projected
Black--Scholes dynamics as the truncation level increases. Since the
forward-time problem is ill-posed, this consistency does not justify
direct propagation of the reduced system; it must still be combined with
regularization.
\end{Remark}

The remaining task is to solve the finite-dimensional initial-value
problem \eqref{eq:projected-ode-vector}--\eqref{eq:projected-initial-data}
for the Fourier-mode vector $\mathbf u(t)$. Once $\mathbf u(t)$ is
computed, due to \eqref{3.2}, the approximate solution of the Black--Scholes equation is
recovered by
\begin{equation}\label{3.11}
    u^N(t,S)=\sum_{n=0}^N u_n(t)\ell_n(S).
\end{equation}
In particular, the desired option--price profile at maturity $T$ is approximated by
\begin{equation} \label{3.12}
    u(T,S)\approx u^N(T,S)
    =
    \sum_{n=0}^N u_n(T)\ell_n(S).
\end{equation}
In this work, the main reconstruction is obtained by applying Tikhonov
regularization to the reduced coefficient system. We also test a reduced
PINN solver as a secondary numerical approach. Both solvers are applied to
the reduced ordinary differential system rather than directly to the
original Black--Scholes equation in the price variable.

\begin{Remark}
Even though the Legendre reduction transforms the Black--Scholes equation
into a finite-dimensional system of ordinary differential equations, a
direct solution of this system is still not reliable for the present
forward-time problem. The reason is that the original problem is
ill-posed, and this instability is inherited by the reduced system.

Indeed, solving
\[
    \mathbf u'(t)=C(t)\mathbf u(t),
    \qquad
    \mathbf u(0)=\mathbf u^0,
\]
directly propagates the initial coefficient vector $\mathbf u^0$ forward
in time. In our computations, this direct propagation produced unstable
reconstructions. This indicates that small errors in the initial
coefficients can be strongly amplified as time increases. Equivalently,
the discrete forward propagator for the reduced system may contain large
amplification factors. For this reason, the main method in this paper solves the reduced system
using Legendre--Tikhonov regularization. The reduced PINN solver is used
only as an additional computational comparison.
\end{Remark}

\begin{Remark}
The method analyzed below applies Tikhonov regularization to the
dimension-reduced Legendre coefficient system and is called the
dimension-reduced Legendre--Tikhonov method. We also report results from a
reduced PINN solver for comparison. The corresponding computational
procedures are summarized in Algorithm~\ref{alg:tikhonov-continuous} and
Algorithm~\ref{alg:pinn-reduced-system}, respectively.
\end{Remark}

\begin{Remark}
The choice of the shifted Legendre basis is motivated by several reasons.
First, after truncating the asset-price domain to $(0,S_{\rm max})$, the
problem is posed on a finite interval, and Legendre polynomials provide a
natural complete orthogonal basis in $L^2(0,S_{\rm max})$ after a simple
affine change of variables. Second, unlike Fourier bases, Legendre
polynomials do not impose periodicity in the price variable, which would
be artificial for option-price profiles. Third, the orthonormality of the
basis leads to a simple projection formula for the expansion coefficients
and allows the Black--Scholes equation to be converted into a finite
system of ordinary differential equations in time. Finally, truncating the
Legendre expansion retains only the low-order modes, which act as a
spectral cutoff and help suppress the highly oscillatory components that
are responsible for the strongest instability of the forward-time problem.
\end{Remark}

\section{Legendre--Tikhonov regularization and reduced PINN comparison}\label{sec4}

We now describe the main Legendre--Tikhonov solver for the reduced
coefficient system \eqref{eq:projected-ode-vector}. We also describe a
reduced PINN solver that will be used as a secondary numerical comparison.

For a fixed truncation number $N$, we
assume that the measured initial profile is a noisy version of the exact
initial profile $u^0$. More precisely, we write $u_\delta^0$ for the
noisy data and assume, in the analysis, that
\[
    \|u_\delta^0-u^0\|_{L^2(0,S_{\rm max})}\leq \delta .
\]
The corresponding projected noisy initial vector is defined by
\[
{\bf u}^0_\delta =
\begin{bmatrix}
   \displaystyle \int_0^{S_{\rm max}} u_\delta^0(S)\ell_m(S)\,dS
\end{bmatrix}_{m=0}^N .
\]
When the exact data $u^0$ are used, the corresponding projected vector is
denoted by
\[
{\bf u}^0 =
\begin{bmatrix}
   \displaystyle \int_0^{S_{\rm max}} u^0(S)\ell_m(S)\,dS
\end{bmatrix}_{m=0}^N .
\]
Since $\{\ell_m\}_{m=0}^N$ is orthonormal in $L^2(0,S_{\rm max})$, the
projected data satisfy
\[
    |{\bf u}^0_\delta-{\bf u}^0|
    \leq
    \|u_\delta^0-u^0\|_{L^2(0,S_{\rm max})}
    \leq \delta .
\]
In the numerical experiments, $u_\delta^0$ is generated by the
multiplicative noise model described in Section~\ref{sec:data}.
In the analysis, only the above projected noise bound is used. We write
$C_N(t)$ for the $(N+1)\times(N+1)$ matrix obtained from the Legendre
projection. In the algorithms, $C(t)$ denotes the same matrix when the
value of $N$ is fixed.

\subsection{Dimension-reduced Legendre--Tikhonov method}

The first approach combines the Legendre dimension reduction with
Tikhonov regularization. After projecting the solution onto a finite
Legendre basis in the price variable $S$, the original partial
differential equation in $(t,S)$ is transformed into a finite-dimensional
system of ordinary differential equations in time for the coefficient
vector $\mathbf u(t)$, namely \eqref{eq:projected-ode-vector}. We then compute $\mathbf u(t)$ by minimizing a
least-squares residual of this reduced system, together with a penalty
enforcing the projected initial data and an $H^2$ regularization term. For this reason, we call the method the dimension-reduced
Legendre--Tikhonov method. The regularization parameter $\alpha$ is chosen
with the guidance of the L-curve criterion described in
Subsection~\ref{alphaN}. This criterion balances the residual of the
reduced ordinary differential system and the projected initial-data
constraint against the $H^2$ regularization norm. After the regularized coefficient vector $\mathbf u_\alpha(t)$ is
computed, the approximate Black--Scholes solution is reconstructed from
the Legendre expansion as in \eqref{3.11}
\[
    u_\alpha^N(t,S)=\sum_{n=0}^N u_{\alpha,n}(t)\ell_n(S).
\]
The option price profile at the expiration time $T$ is then obtained as in \eqref{3.12}
\[
    u_\alpha^N(T,S)=\sum_{n=0}^N u_{\alpha,n}(T)\ell_n(S).
\]
This procedure is summarized in Algorithm \ref{alg:tikhonov-continuous}.

\begin{algorithm}[ht]
\caption{dimension-reduced Legendre--Tikhonov method}
\label{alg:tikhonov-continuous}
\begin{algorithmic}[1]
\State Choose a regularization parameter $\alpha$ and a cutoff number $N$ as in Subsection \ref{alphaN}.
\State Find $\mathbf u_\alpha\in H^2(0,T;\mathbb R^{N+1})$ as the minimizer of $J_{\alpha,N}^\delta(\mathbf u)$ defined in \eqref{4.1}.

\State Reconstruct the approximate Black--Scholes solution by
\[
    u_{\alpha}^N(t,S)
    =
    \sum_{n=0}^N u_{\alpha,n}(t)\ell_n(S),
    \qquad 0\leq t\leq T.
\]

\State Compute the desired option price profile at maturity $T$ by
\[
    u(T,S)\approx u_\alpha^N(T,S)
    =
    \sum_{n=0}^N u_{\alpha,n}(T)\ell_n(S).
\]

\State \Return $u_\alpha^N(t,S)$ and $u_\alpha^N(T,S)$.
\end{algorithmic}
\end{algorithm}

\subsection{Stability and convergence of the Legendre--Tikhonov method}

We now justify the dimension-reduced Legendre--Tikhonov method for a fixed
truncation number $N$. For $\alpha>0$, define
\begin{equation}\label{4.1}
    J_{\alpha,N}^\delta({\bf v})
    =
    \int_0^T
    |{\bf v}'(t)-C_N(t){\bf v}(t)|^2\,dt
    +
    |{\bf v}(0)-{\bf u}^0_\delta|^2      
    +
    \alpha \|{\bf v}\|_{H^2(0,T;\mathbb R^{N+1})}^2 ,
\end{equation}
for ${\bf v}\in H^2(0,T;\mathbb R^{N+1})$.

\begin{Theorem}[Stability and convergence for fixed $N$]\label{thm:tikhonov-fixed-N}
Assume that $C_N\in W^{1,\infty}(0,T;\mathbb R^{(N+1)\times(N+1)})$.
Then, for every $\alpha>0$, the functional $J_{\alpha,N}^\delta$ has a
unique minimizer
\[
    {\bf u}_{\alpha,N}^\delta
    \in H^2(0,T;\mathbb R^{N+1}).
\]
Moreover, the minimizer depends continuously on the projected data. More
precisely, if ${\bf u}_{\alpha,N}^{\delta,1}$ and
${\bf u}_{\alpha,N}^{\delta,2}$ are the minimizers corresponding to
projected initial vectors $({\bf u}^0_\delta)^1$ and $({\bf u}^0_\delta)^2$,
then
\[
    \|{\bf u}_{\alpha,N}^{\delta,1}
      -
      {\bf u}_{\alpha,N}^{\delta,2}\|_{H^2(0,T;\mathbb R^{N+1})}
    \le
    \frac{1}{2\sqrt{\alpha}}
    |({\bf u}^0_\delta)^1-({\bf u}^0_\delta)^2|.
\]

Let ${\bf u}_N^\dagger$ be the exact solution of the reduced system
\[
    ({\bf u}_N^\dagger)'(t)=C_N(t){\bf u}_N^\dagger(t),
    \qquad
    {\bf u}_N^\dagger(0)={\bf u}^0 .
\]
Assume that $|{\bf u}^0_\delta-{\bf u}^0|\le \delta$. If
$\alpha=\alpha(\delta)$ satisfies
\[
    \alpha(\delta)\to 0,
    \qquad
    \frac{\delta^2}{\alpha(\delta)}\to 0
    \quad \text{as } \delta\to 0,
\]
then
\[
    {\bf u}_{\alpha(\delta),N}^\delta
    \to
    {\bf u}_N^\dagger
    \quad
    \text{strongly in } H^2(0,T;\mathbb R^{N+1}).
\]
In particular,
\[
    {\bf u}_{\alpha(\delta),N}^\delta(T)
    \to
    {\bf u}_N^\dagger(T)
    \quad
    \text{in } \mathbb R^{N+1}.
\]
\end{Theorem}

\begin{proof}
Since $C_N\in W^{1,\infty}(0,T;\mathbb R^{(N+1)\times(N+1)})$, the map
${\bf v}\mapsto {\bf v}'-C_N(t){\bf v}$ is bounded from
$H^2(0,T;\mathbb R^{N+1})$ to $L^2(0,T;\mathbb R^{N+1})$. The trace map
${\bf v}\mapsto {\bf v}(0)$ is also bounded from
$H^2(0,T;\mathbb R^{N+1})$ to $\mathbb R^{N+1}$. Hence
$J_{\alpha,N}^\delta$ is continuous on
$H^2(0,T;\mathbb R^{N+1})$.

Moreover, the term
$\alpha\|{\bf v}\|_{H^2(0,T;\mathbb R^{N+1})}^2$ makes
$J_{\alpha,N}^\delta$ coercive and strictly convex. Therefore
$J_{\alpha,N}^\delta$ has a unique minimizer
${\bf u}_{\alpha,N}^\delta\in H^2(0,T;\mathbb R^{N+1})$.

We next prove the stability estimate. Define the bounded linear operator
\[
    \mathcal A_N:
    H^2(0,T;\mathbb R^{N+1})
    \to
    L^2(0,T;\mathbb R^{N+1})\times \mathbb R^{N+1}
\]
by
\[
    \mathcal A_N{\bf v}
    =
    \big({\bf v}'-C_N(t){\bf v},{\bf v}(0)\big).
\]
We equip $L^2(0,T;\mathbb R^{N+1})\times\mathbb R^{N+1}$ with the norm
\[
    \|({\bf f},{\bf a})\|^2
    =
    \int_0^T |{\bf f}(t)|^2\,dt
    +
    |{\bf a}|^2 .
\]
Thus, for ${\bf v}\in H^2(0,T;\mathbb R^{N+1})$,
\[
    \|\mathcal A_N{\bf v}\|^2
    =
    \int_0^T |{\bf v}'(t)-C_N(t){\bf v}(t)|^2\,dt
    +
    |{\bf v}(0)|^2 .
\]
Let $\mathcal A_N^*$ denote the Hilbert-space adjoint of
$\mathcal A_N$, and let $I$ denote the identity operator on
$H^2(0,T;\mathbb R^{N+1})$. With these notations, the Tikhonov functional can
be written as
\[
    J_{\alpha,N}^\delta({\bf v})
    =
    \|\mathcal A_N{\bf v}-(0,{\bf u}^0_\delta)\|^2
    +
    \alpha\|{\bf v}\|_{H^2(0,T;\mathbb R^{N+1})}^2 .
\]
Thus the minimizer satisfies the normal equation
\[
    (\mathcal A_N^*\mathcal A_N+\alpha I)
    {\bf u}_{\alpha,N}^\delta
    =
    \mathcal A_N^*(0,{\bf u}^0_\delta).
\]

Let ${\bf u}_{\alpha,N}^{\delta,1}$ and
${\bf u}_{\alpha,N}^{\delta,2}$ be the minimizers corresponding to the
projected noisy data $({\bf u}^0_\delta)^1$ and
$({\bf u}^0_\delta)^2$. Subtracting the two normal equations gives
\begin{equation*}
    (\mathcal A_N^*\mathcal A_N+\alpha I)
    \left(
    {\bf u}_{\alpha,N}^{\delta,1}
    -
    {\bf u}_{\alpha,N}^{\delta,2}
    \right)   
    =
    \mathcal A_N^*
    \big(0,({\bf u}^0_\delta)^1-({\bf u}^0_\delta)^2\big).
\end{equation*}
Taking the inner product with
${\bf u}_{\alpha,N}^{\delta,1}
-
{\bf u}_{\alpha,N}^{\delta,2}$
in $H^2(0,T;\mathbb R^{N+1})$, we obtain
\begin{multline*}
    \left\|
    \mathcal A_N
    \left(
    {\bf u}_{\alpha,N}^{\delta,1}
    -
    {\bf u}_{\alpha,N}^{\delta,2}
    \right)
    \right\|^2
    +
    \alpha
    \left\|
    {\bf u}_{\alpha,N}^{\delta,1}
    -
    {\bf u}_{\alpha,N}^{\delta,2}
    \right\|_{H^2(0,T;\mathbb R^{N+1})}^2  \\
    \leq
    |({\bf u}^0_\delta)^1-({\bf u}^0_\delta)^2|
    \left\|
    \mathcal A_N
    \left(
    {\bf u}_{\alpha,N}^{\delta,1}
    -
    {\bf u}_{\alpha,N}^{\delta,2}
    \right)
    \right\| .
\end{multline*}
Using $a^2+b^2\geq 2ab$ with
\[
    a=
    \left\|
    \mathcal A_N
    \left(
    {\bf u}_{\alpha,N}^{\delta,1}
    -
    {\bf u}_{\alpha,N}^{\delta,2}
    \right)
    \right\|,
    \qquad
    b=
    \sqrt{\alpha}
    \left\|
    {\bf u}_{\alpha,N}^{\delta,1}
    -
    {\bf u}_{\alpha,N}^{\delta,2}
    \right\|_{H^2(0,T;\mathbb R^{N+1})},
\]
we get
\begin{multline*}
    2\sqrt{\alpha}
    \left\|
    \mathcal A_N
    \left(
    {\bf u}_{\alpha,N}^{\delta,1}
    -
    {\bf u}_{\alpha,N}^{\delta,2}
    \right)
    \right\|
    \left\|
    {\bf u}_{\alpha,N}^{\delta,1}
    -
    {\bf u}_{\alpha,N}^{\delta,2}
    \right\|_{H^2(0,T;\mathbb R^{N+1})}  \\
    \leq
    |({\bf u}^0_\delta)^1-({\bf u}^0_\delta)^2|
    \left\|
    \mathcal A_N
    \left(
    {\bf u}_{\alpha,N}^{\delta,1}
    -
    {\bf u}_{\alpha,N}^{\delta,2}
    \right)
    \right\| .
\end{multline*}
If the last factor involving $\mathcal A_N$ is zero, then the desired
estimate is immediate. Otherwise, dividing by this factor gives
\[
    \|{\bf u}_{\alpha,N}^{\delta,1}
      -
      {\bf u}_{\alpha,N}^{\delta,2}\|_{H^2(0,T;\mathbb R^{N+1})}
    \leq
    \frac{1}{2\sqrt{\alpha}}
    |({\bf u}^0_\delta)^1-({\bf u}^0_\delta)^2|.
\]
This proves the stability estimate.

We now prove convergence. Since ${\bf u}_N^\dagger$ solves the exact
reduced system, we have
\[
    ({\bf u}_N^\dagger)'(t)-C_N(t){\bf u}_N^\dagger(t)=0,
    \qquad
    {\bf u}_N^\dagger(0)={\bf u}^0 .
\]
By the minimizing property of ${\bf u}_{\alpha,N}^\delta$,
\begin{multline}\label{eq:tikhonov-min-bound}
    \int_0^T
    |({\bf u}_{\alpha,N}^\delta)'(t)
    -C_N(t){\bf u}_{\alpha,N}^\delta(t)|^2\,dt
    +
    |{\bf u}_{\alpha,N}^\delta(0)-{\bf u}^0_\delta|^2  \\
    +
    \alpha
    \|{\bf u}_{\alpha,N}^\delta\|_{H^2(0,T;\mathbb R^{N+1})}^2
    \leq
    |{\bf u}^0_\delta-{\bf u}^0|^2
    +
    \alpha
    \|{\bf u}_N^\dagger\|_{H^2(0,T;\mathbb R^{N+1})}^2 .
\end{multline}
Since $|{\bf u}^0_\delta-{\bf u}^0|\leq\delta$, we obtain
\begin{equation}\label{eq:H2-bound-Tikhonov}
    \|{\bf u}_{\alpha,N}^\delta\|_{H^2(0,T;\mathbb R^{N+1})}^2
    \leq
    \frac{\delta^2}{\alpha}
    +
    \|{\bf u}_N^\dagger\|_{H^2(0,T;\mathbb R^{N+1})}^2 .
\end{equation}
Let $\alpha=\alpha(\delta)$ satisfy
$\alpha(\delta)\to0$ and $\delta^2/\alpha(\delta)\to0$. Then
\eqref{eq:H2-bound-Tikhonov} implies that
${\bf u}_{\alpha(\delta),N}^\delta$ is bounded in
$H^2(0,T;\mathbb R^{N+1})$.

Returning to \eqref{eq:tikhonov-min-bound}, its right-hand side tends to
zero because $|{\bf u}^0_\delta-{\bf u}^0|\leq\delta$, $\delta\to0$, and
$\alpha(\delta)\to0$. Since all terms on the left-hand side are
nonnegative, we get
\[
    ({\bf u}_{\alpha(\delta),N}^\delta)'
    -
    C_N(t){\bf u}_{\alpha(\delta),N}^\delta
    \to 0
    \quad
    \text{in } L^2(0,T;\mathbb R^{N+1}),
\]
and
\[
    {\bf u}_{\alpha(\delta),N}^\delta(0)
    -
    {\bf u}^0_\delta
    \to 0
    \quad
    \text{in } \mathbb R^{N+1}.
\]
Since $|{\bf u}^0_\delta-{\bf u}^0|\leq\delta\to0$, it follows that
\[
    {\bf u}_{\alpha(\delta),N}^\delta(0)
    \to
    {\bf u}^0
    \quad
    \text{in } \mathbb R^{N+1}.
\]

We now identify the weak limit. Since
${\bf u}_{\alpha(\delta),N}^\delta$ is bounded in
$H^2(0,T;\mathbb R^{N+1})$, every sequence $\delta_j\to0$ has a
subsequence, still denoted by $\delta_j$, such that
\[
    {\bf u}_{\alpha(\delta_j),N}^{\delta_j}
    \rightharpoonup
    {\bf w}
    \quad
    \text{weakly in } H^2(0,T;\mathbb R^{N+1}).
\]
The residual convergence and the weak convergence imply
\[
    {\bf w}'(t)-C_N(t){\bf w}(t)=0
    \quad
    \text{in } L^2(0,T;\mathbb R^{N+1}).
\]
The trace convergence gives
\[
    {\bf w}(0)={\bf u}^0 .
\]
Thus ${\bf w}$ solves
\[
    {\bf w}'(t)=C_N(t){\bf w}(t),
    \qquad
    {\bf w}(0)={\bf u}^0 .
\]
By uniqueness of this finite-dimensional initial-value problem,
${\bf w}={\bf u}_N^\dagger$. Hence the whole family satisfies
\[
    {\bf u}_{\alpha(\delta),N}^\delta
    \rightharpoonup
    {\bf u}_N^\dagger
    \quad
    \text{weakly in } H^2(0,T;\mathbb R^{N+1}).
\]

It remains to upgrade weak convergence to strong convergence. From
\eqref{eq:H2-bound-Tikhonov} and $\delta^2/\alpha(\delta)\to0$, we have
\[
    \limsup_{\delta\to0}
    \|{\bf u}_{\alpha(\delta),N}^\delta\|_{H^2(0,T;\mathbb R^{N+1})}
    \leq
    \|{\bf u}_N^\dagger\|_{H^2(0,T;\mathbb R^{N+1})}.
\]
On the other hand, weak lower semicontinuity of the norm gives
\[
    \|{\bf u}_N^\dagger\|_{H^2(0,T;\mathbb R^{N+1})}
    \leq
    \liminf_{\delta\to0}
    \|{\bf u}_{\alpha(\delta),N}^\delta\|_{H^2(0,T;\mathbb R^{N+1})}.
\]
Therefore
\[
    \|{\bf u}_{\alpha(\delta),N}^\delta\|_{H^2(0,T;\mathbb R^{N+1})}
    \to
    \|{\bf u}_N^\dagger\|_{H^2(0,T;\mathbb R^{N+1})}.
\]
In a Hilbert space, weak convergence together with convergence of the
norms implies strong convergence. Hence
\[
    {\bf u}_{\alpha(\delta),N}^\delta
    \to
    {\bf u}_N^\dagger
    \quad
    \text{strongly in } H^2(0,T;\mathbb R^{N+1}).
\]
Finally, the convergence at $t=T$ follows from the continuous embedding
$H^2(0,T;\mathbb R^{N+1})\hookrightarrow C^1([0,T];\mathbb R^{N+1})$.
\end{proof}

\begin{Corollary}[Convergence of the reduced terminal profile]
Under the assumptions of Theorem~\ref{thm:tikhonov-fixed-N}, the
reconstruction
\[
    u_{\alpha,N}^\delta(T,S)
    =
    \sum_{n=0}^N u_{\alpha,N,n}^\delta(T)\ell_n(S)
\]
converges in $L^2(0,S_{\rm max})$ to the exact reduced terminal profile
\[
    u_N^\dagger(T,S)
    =
    \sum_{n=0}^N u_{N,n}^\dagger(T)\ell_n(S).
\]
Indeed, by the orthonormality of $\{\ell_n\}_{n=0}^N$ in
$L^2(0,S_{\rm max})$,
\[
    \|u_{\alpha,N}^\delta(T,\cdot)-u_N^\dagger(T,\cdot)\|_{L^2(0,S_{\rm max})}
    =
    |{\bf u}_{\alpha,N}^\delta(T)-{\bf u}_N^\dagger(T)|
    \to 0 .
\]
\end{Corollary}

The preceding result is a fixed-$N$ regularization theorem. The error with
respect to the original Black--Scholes solution also contains the
Legendre truncation error:
\begin{multline*}
\|u_{\alpha,N}^{\delta}(T,\cdot)-u(T,\cdot)\|_{L^2(0,S_{\rm max})}
\le
\|u_{\alpha,N}^{\delta}(T,\cdot)-u_N^\dagger(T,\cdot)\|_{L^2(0,S_{\rm max})} \\
+
\|u_N^\dagger(T,\cdot)-u(T,\cdot)\|_{L^2(0,S_{\rm max})}.
\end{multline*}
The first term is controlled by the theorem above, while the second term
is the Legendre truncation error.

\subsection{Reduced PINN solver for comparison}
\label{subsec:pinn-secondary}

We also consider a reduced PINN solver as a secondary computational
approach for the Legendre coefficient system. This solver is not used in
the convergence analysis. Instead, it is included to examine whether a
neural-network parametrization can provide a useful numerical solver after
the price-dimensional Legendre reduction.

Let
\[
    {\bf u}(t)=\big(u_0(t),u_1(t),\ldots,u_N(t)\big)^\top
\]
be the coefficient vector in the reduced system
\[
    {\bf u}'(t)=C_N(t){\bf u}(t).
\]
We approximate ${\bf u}(t)$ by a neural network
\[
    \mathcal N_\theta(t)
    =
    \big(
        u_{\theta,0}(t),
        u_{\theta,1}(t),
        \ldots,
        u_{\theta,N}(t)
    \big)^\top .
\]
The network is trained by minimizing the residual of the reduced ordinary
differential system together with the mismatch with the projected noisy
initial vector ${\bf u}^0_\delta$. The derivatives of
$\mathcal N_\theta$ with respect to $t$ are computed by automatic
differentiation.

The reduced PINN solver should be viewed as a neural-network
parametrization of the reduced coefficient problem. Unlike the
Legendre--Tikhonov method, we do not prove a separate stability or
convergence theorem for the trained PINN. Its role in this paper is
numerical: it provides a secondary solver for comparison with the
Legendre--Tikhonov method, which is the theoretically justified
regularized reconstruction method.

After training, let $\theta^\ast$ denote the optimized parameters. The
approximate solution of the Black--Scholes equation is reconstructed by
\[
    u_{\theta^\ast}^N(t,S)
    =
    \sum_{n=0}^N
    u_{\theta^\ast,n}(t)\ell_n(S),
\]
and the terminal option price profile is obtained by evaluating this expression at
$t=T$. This procedure is summarized in
Algorithm~\ref{alg:pinn-reduced-system}.

\begin{algorithm}[h]
\caption{Reduced PINN solver for the Legendre coefficient system}
\label{alg:pinn-reduced-system}
\begin{algorithmic}[1]

\State Choose a neural-network architecture and training parameters as
specified in the numerical section.

\State Minimize the reduced PINN loss described in the numerical section
with respect to $\theta$ using a gradient-based optimizer.

\State Denote the trained network by $\mathcal N_{\theta^\ast}$.

\State Reconstruct the approximate Black--Scholes solution by
\[
    u_{\theta^\ast}^N(t,S)
    =
    \sum_{n=0}^N
    \big(\mathcal N_{\theta^\ast}(t)\big)_n \ell_n(S),
    \qquad 0\leq t\leq T.
\]

\State Compute the desired terminal option price profile by
\[
    u(T,S)\approx u_{\theta^\ast}^N(T,S)
    =
    \sum_{n=0}^N
    \big(\mathcal N_{\theta^\ast}(T)\big)_n \ell_n(S).
\]

\State \Return $u_{\theta^\ast}^N(t,S)$ and $u_{\theta^\ast}^N(T,S)$.

\end{algorithmic}
\end{algorithm}

\section{Numerical study}\label{sec5}

In the numerical experiments, the final time $T$ is varied from test to
test. For each value of $T$, we use the truncated asset-price interval $[0, S_{\rm max}] = [0, 10].$
In all tests, we set the risk-free interest rate
$
    r=0.05.
$
The local volatility is chosen as a simple time-dependent volatility smile,
\begin{equation}\label{eq:numerical-local-volatility}
    \sigma(t,S)
    =
    \sigma_0
    \sqrt{
    1+\eta e^{-t/T}
    \left(\frac{S-S_{\rm ref}}{S_{\rm ref}}\right)^2
    },
    \qquad
    \sigma_0=0.2,\quad
    \eta=0.25,\quad
    S_{\rm ref}=\frac{S_{\max}}{2}.
\end{equation}
This choice is positive and produces a mild volatility smile around the
reference price $S_{\rm ref}$. In this context, the term ``volatility smile"
means that the local volatility is lowest near $S_{\rm ref}$ and increases
as the asset price moves away from $S_{\rm ref}$. The time-dependent
factor makes the smile effect gradually decrease as time increases.

The spatial grid is fixed by taking
\[
    N_S=100,
    \qquad
    \Delta S=\frac{S_{\max}}{N_S}=0.1.
\]
The spatial grid points are then defined by
\[
    S_i=(i-1)\Delta S,\qquad i=1,\ldots,N_S+1,
\]
so that $S_1=0$ and $S_{N_S+1}=S_{\max}$.

For the time discretization, we choose the number of time steps $N_t$ so
that the time step is approximately $5\times 10^{-4}$, namely
\[
    N_t=\operatorname{round}\left(\frac{T}{5\times 10^{-4}}\right),
    \qquad
    \Delta t=\frac{T}{N_t}.
\]
The time grid points are
\[
    t_n=(n-1)\Delta t,\qquad n=1,\ldots,N_t+1,
\]
so that $t_1=0$ and $t_{N_t+1}=T$.

For the above parameters, the sufficient stability condition becomes
\[
    \Delta t
    \left(
        \frac{\sigma^2(t_n,S_i) S_i^2}{(\Delta S)^2}
        +
        \frac{rS_i}{\Delta S}
        +
        r
    \right)
    \leq 1.
\]
Using the conservative bounds $S_i\leq S_{\max}=10$ and
\[
    \sigma(t,S)\leq
    \sigma_{\max}
    =
    0.2\sqrt{1+0.25}
    \approx 0.2236,
\]
we obtain
\[
    \Delta t
    \left(
        \frac{\sigma_{\max}^2 S_{\max}^2}{(\Delta S)^2}
        +
        \frac{rS_{\max}}{\Delta S}
        +
        r
    \right)
    =
    5\times 10^{-4}
    \left(
        \frac{0.2236^2\cdot 10^2}{0.1^2}
        +
        \frac{0.05\cdot 10}{0.1}
        +
        0.05
    \right)
    \approx 0.2525<1.
\]
Hence, the explicit data-generation scheme satisfies the stated sufficient stability condition for all tested values of $T$, since the time step
$\Delta t$ is kept approximately fixed at $5\times 10^{-4}$.

\subsection{Data generation}\label{sec:data}

To assess the performance of the proposed method, we generate synthetic
data from a prescribed terminal profile. Let $\Phi$ be the terminal payoff
on the truncated price interval $[0,S_{\max}]$. We first solve the
standard terminal-value Black--Scholes problem
\begin{equation}
    u_t+\frac12\sigma^2(t,S)S^2u_{SS}+rS u_S-ru=0,\qquad
    (t,S)\in(0,T)\times(0,S_{\max}),
\end{equation}
with terminal condition
\[
    u(T,S)=\Phi(S).
\]
The resulting numerical solution is denoted by $u_{\rm true}(t,S)$.

We compute this reference solution backward in time by an explicit
finite-difference scheme. The terminal condition is imposed by
\[
    u_i^{N_t+1}=\Phi(S_i),
    \qquad i=1,\ldots,N_S+1.
\]
Then, for $n=N_t,N_t-1,\ldots,1$, the values $u_i^n$ are computed from the
already known values at time level $t_{n+1}$. For each interior point
$i=2,\ldots,N_S$, we use the explicit formula

\begin{multline}\label{eq:explicit-direct-update}
    u_i^n
    =
    u_i^{n+1}
    +
    \Delta t\,
    \frac12\sigma^2(t_{n+1},S_i)S_i^2
    \frac{u_{i+1}^{n+1}-2u_i^{n+1}+u_{i-1}^{n+1}}{(\Delta S)^2}
    \\
    +
    \Delta t\, rS_i
    \frac{u_{i+1}^{n+1}-u_i^{n+1}}{\Delta S}
    -
    r\Delta t\,u_i^{n+1},
    \qquad i=2,\ldots,N_S,
\end{multline}
to compute the interior values of $u^n.$
After these interior values are computed, the endpoint values are obtained
by linear extrapolation:
\begin{equation}\label{eq:explicit-direct-boundary}
    u_1^n=2u_2^n-u_3^n,
    \qquad
    u_{N_S+1}^n=2u_{N_S}^n-u_{N_S-1}^n.
\end{equation}
Repeating this process until $n=1$ gives the computed profile at $t=0$.
We then define the initial data for the forward-time reconstruction
problem by
\[
    u^0(S_i)=u_{\rm true}(0,S_i)=u_i^1,
    \qquad i=1,\ldots,N_S+1.
\]
Thus, the synthetic data are generated consistently from the same
Black--Scholes model used in the reconstruction. In the reconstruction
step, only the initial profile, or its noisy version, is used as input,
and the goal is to recover the option price profile at maturity $u(T,S)$.

To model measurement error, we add multiplicative noise to the initial
profile. Given a noise level $\delta > 0$, the noisy data are defined by
\begin{equation}\label{noise}
    u_\delta^{0}(S_i)
    =
    u^0(S_i)\left(1+\delta \xi_i \right) \qquad i=1,\ldots,N_S+1,
\end{equation}
where $\{\xi_i\}$ are independent random variables uniformly distributed
in $[-1,1]$.

\subsection{Numerical implementation details}

In the MATLAB implementation of Algorithm~\ref{alg:tikhonov-continuous},
the reduced ordinary differential system is discretized on the time grid.
The unknown is the vector containing all coefficient values
\[
    \mathbf U
    =
    \big(
    \mathbf u(t_1)^\top,
    \mathbf u(t_2)^\top,
    \ldots,
    \mathbf u(t_{N_t+1})^\top
    \big)^\top .
\]
The residual of the reduced system $\mathbf u'(t)-C(t)\mathbf u(t)$ is
approximated by a first-order finite difference in time. The initial
condition is enforced by adding the term
$|\mathbf u(t_1)-\mathbf u^0_\delta|^2$. The
$H^2(0,T)$-regularization term is discretized using the values of
$\mathbf u$, the first finite difference in time, and the second finite
difference in time. Therefore, the continuous Tikhonov minimization is
converted into a finite-dimensional least-squares problem of the form
\[
    \min_{\mathbf U}
    \left\|
        A_{\rm aug}\mathbf U-\mathbf b_{\rm aug}
    \right\|_2^2,
\]
where
\[
    A_{\rm aug}
    =
    \begin{pmatrix}
        A_{\rm ode}\\
        A_{\rm init}\\
        \sqrt{\alpha}A_{H^2}
    \end{pmatrix},
    \qquad
    \mathbf b_{\rm aug}
    =
    \begin{pmatrix}
        0\\
        \mathbf u^0_\delta\\
        0
    \end{pmatrix}.
\]
Here $A_{\rm ode}$ represents the discretized reduced ordinary
differential system, $A_{\rm init}$ represents the projected initial
condition, and $A_{H^2}$ represents the discrete $H^2$-regularization
operator. In the computations, the augmented least-squares problem is
solved in MATLAB using the backslash operator applied to
$A_{\rm aug}\mathbf U\approx \mathbf b_{\rm aug}$.

For the reduced PINN computations, we use a fully connected feed-forward
network with one input, three hidden layers, 50 neurons in each hidden
layer, and $N+1$ outputs. The activation function is $\tanh$. The network
output is denoted by $\mathcal N_\theta(t)$ and is trained using the
discrete loss
\begin{equation}\label{PINN}
    \mathcal L_c(\theta)
    =
    \frac{1}{N_c}
    \sum_{k=1}^{N_c}
    \left|
        \frac{d}{dt}\mathcal N_\theta(t_k)
        -
        C(t_k)\mathcal N_\theta(t_k)
    \right|^2
    +
    \left|
        \mathcal N_\theta(0)-{\bf u}^0_\delta
    \right|^2 .
\end{equation}
Here $\{t_k\}_{k=1}^{N_c}\subset(0,T)$ are collocation points. The
derivative $d\mathcal N_\theta/dt$ is computed by automatic
differentiation. In all reported PINN computations, we use $N_c=1000$
collocation points, train for 5000 epochs, and use learning rate
$10^{-3}$. The reduced PINN results are reported only as a secondary
computational comparison with the Legendre--Tikhonov reconstruction.

\subsection{Numerical examples}

We now show several numerical results obtained by Algorithms \ref{alg:tikhonov-continuous} and \ref{alg:pinn-reduced-system}. The Legendre--Tikhonov method is the main reconstruction method, while the reduced PINN solver is included as a secondary numerical comparison. After presenting all tests, we summarize the relative $L^2$ errors in Table~\ref{tab:numerical-errors}.

{\bf Test 1: smooth compactly supported profile.}
In the first test, we use a smooth compactly supported terminal profile.
This example is not meant to model a standard financial payoff. Instead,
it provides a benchmark for evaluating the stability and accuracy of the
reconstruction methods when the target profile is regular.

We set the final time to $T=1$ and define the true terminal option-price
profile by
\[
    \Phi(S)
    =
    \begin{cases}
    \displaystyle
    \exp\left(1-\frac{1}{1-\left(\frac{S-5}{2}\right)^2}\right),
    &
    |S-5|<2,\\[2mm]
    0,
    &
    |S-5|\geq 2.
    \end{cases}
\]
The function is centered at $S=5$, has maximum value equal to one, and
vanishes outside the interval $(3,7)$. Since this profile is smooth and
compactly supported, it gives a clean test case for studying the effect
of noise and regularization without the additional complications caused
by nonsmooth payoff kinks.

Figure~\ref{fig:test1} presents the numerical results for Test 1
with $10\%$ and $35\%$ noise added to the initial data. For each noise
level, we show the noisy observed initial profile and compare the exact
option price profile at $T$ with the reconstructions obtained by the Tikhonov
method in Algorithm~\ref{alg:tikhonov-continuous} and the reduced PINN
solver in Algorithm~\ref{alg:pinn-reduced-system}.

\begin{figure}[ht]
    \centering
    \subfloat[Initial data, $\delta=10\%$]{
        \includegraphics[width=0.44\textwidth]{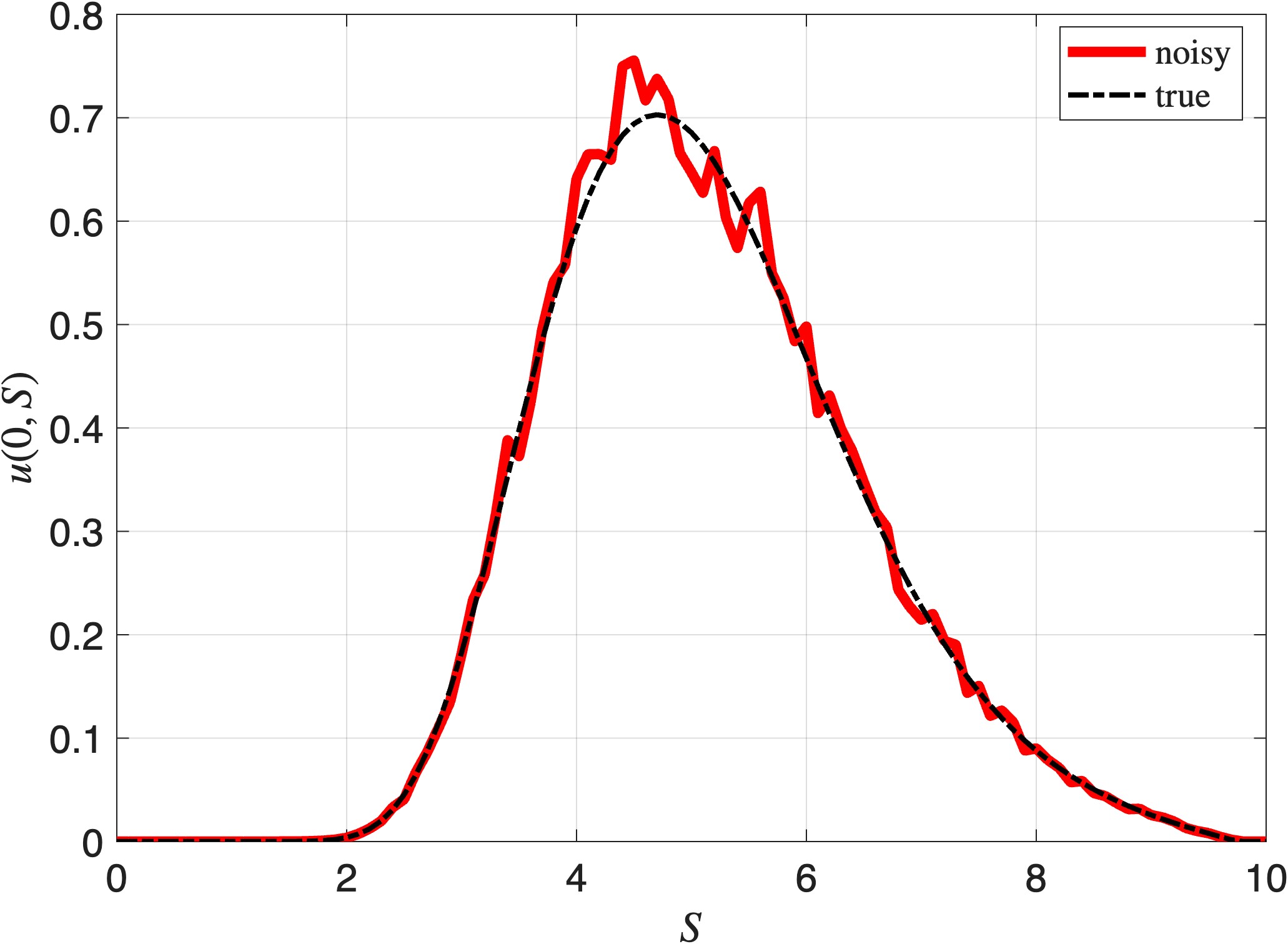}
    }
    \hfill
    \subfloat[Reconstruction of $u(T,S)$ from $u^0_\delta$, $\delta=10\%$]{
        \includegraphics[width=0.44\textwidth]{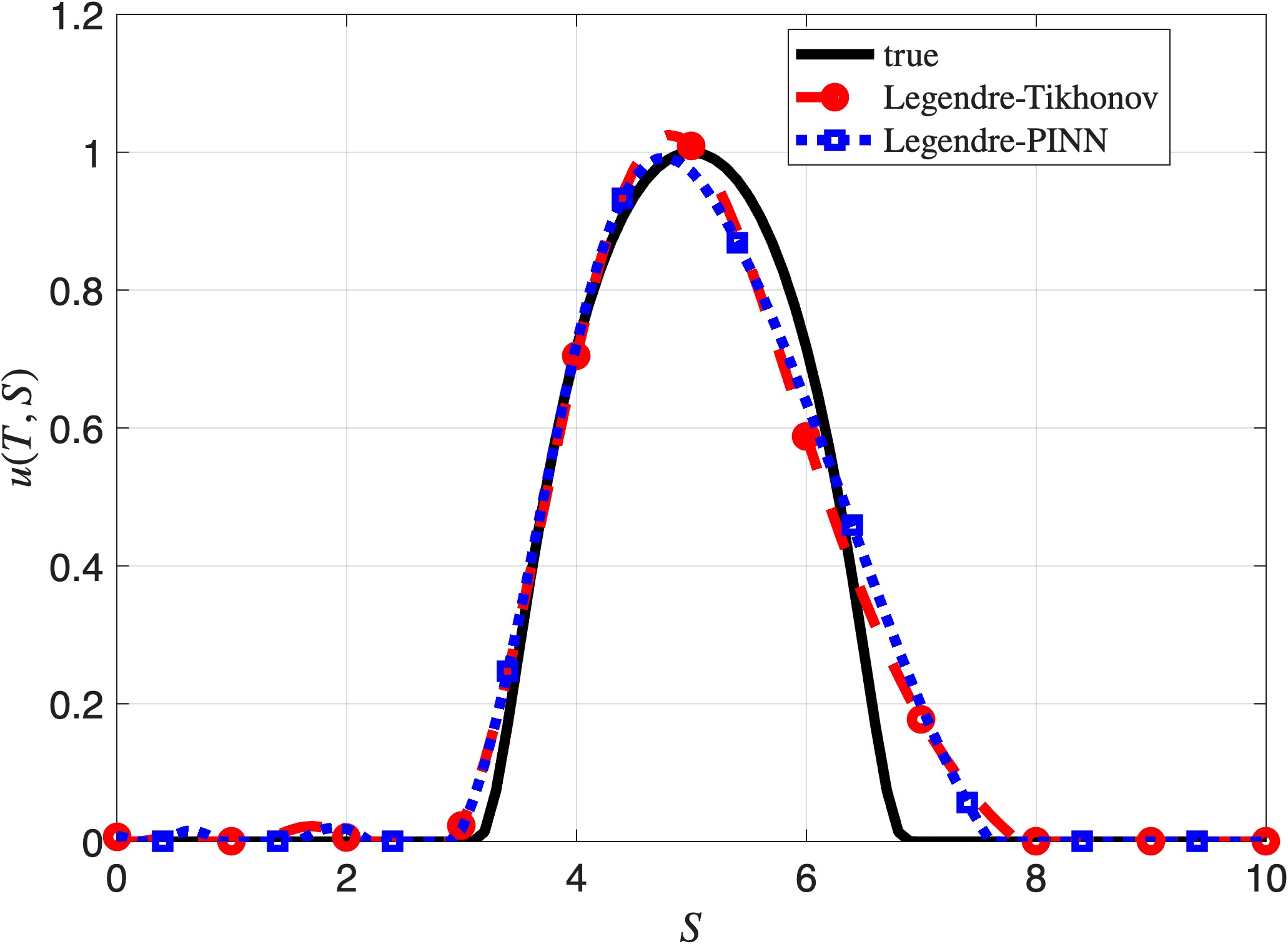}
    }
    
    \subfloat[Initial data, $\delta=35\%$]{
        \includegraphics[width=0.44\textwidth]{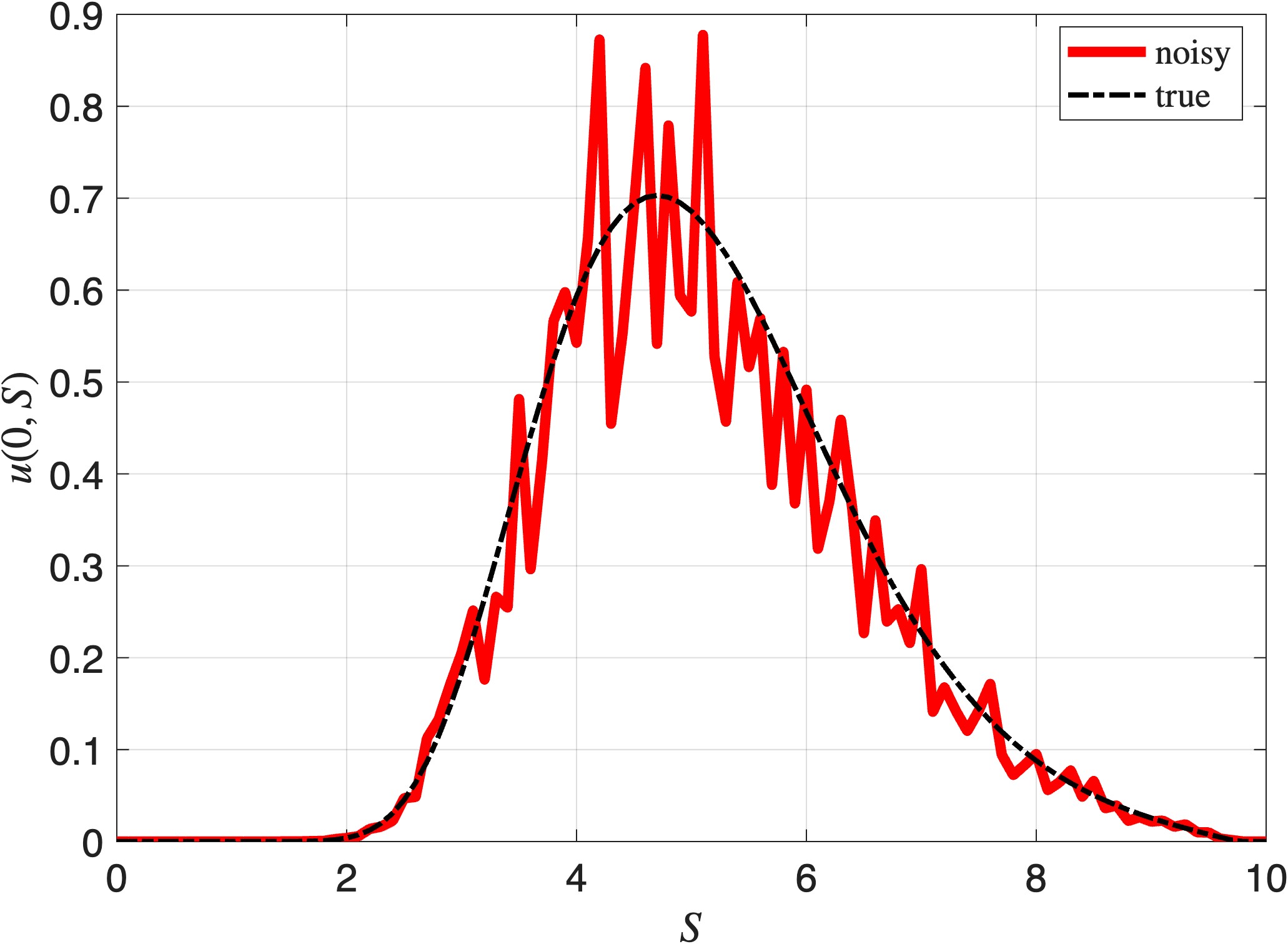}
    }
    \hfill
    \subfloat[Reconstruction of $u(T,S)$ from $u^0_\delta$, $\delta=35\%$]{
        \includegraphics[width=0.44\textwidth]{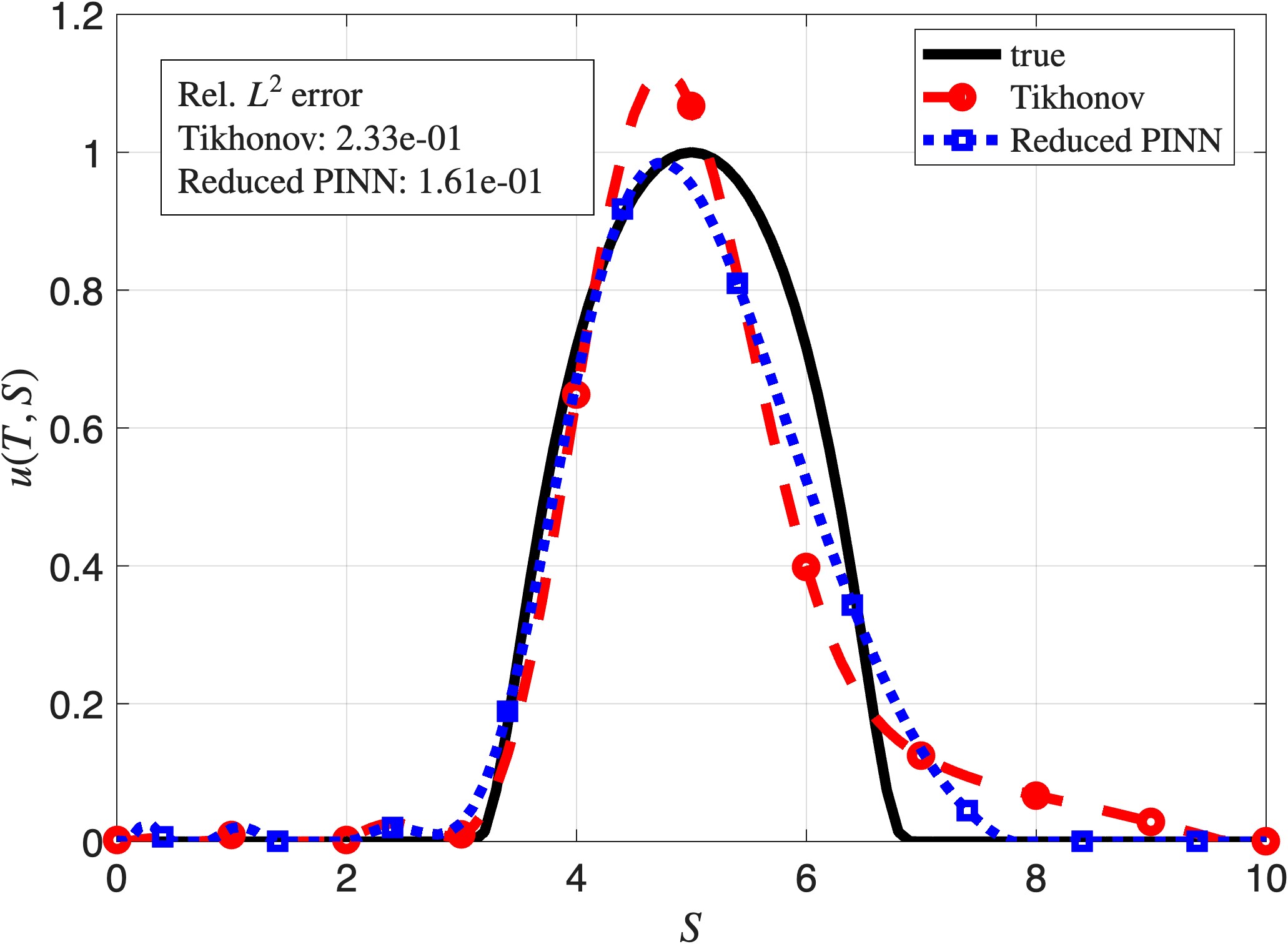}
    }
\caption{
Numerical results for Test 1 with final time $T=1$ and noisy initial data
generated according to \eqref{noise}. For $\delta=10\%$, the relative
$L^2$ errors are $14.56\%$ for the Tikhonov method and $15.40\%$ for the
reduced PINN solver. For $\delta=35\%$, the corresponding errors are
$23.27\%$ and $16.09\%$, respectively.
}
    \label{fig:test1}
\end{figure}

Figure~\ref{fig:test1} demonstrates that both proposed algorithms
are able to recover the main structure of the option--price profile at maturity $T$ from
noisy initial data when $T=1$. For the $10\%$ noise case, the noisy
initial profile is visibly perturbed, but both the Tikhonov method and
the reduced PINN solver reconstruct the future profile accurately. In
particular, the location of the main peak, the compact support, and the
overall shape of $u(T,S)$ are well captured. The Tikhonov method gives a
relative $L^2$ error of $14.56\%$, while the reduced PINN solver gives a
comparable relative $L^2$ error of $15.40\%$. Thus, for this moderate
noise level, the two methods have similar accuracy, with the Tikhonov
reconstruction being slightly closer to the exact profile near the peak.

For the more challenging $35\%$ noise case, the initial data contain
strong oscillations. Nevertheless, the Legendre--Tikhonov method and the reduced PINN solver still recover the dominant
future-time structure. The reconstructions correctly identify the main
support region and the principal peak of the exact solution, although the
accuracy deteriorates compared with the $10\%$ case. The Tikhonov method
produces a relative $L^2$ error of $23.27\%$, while the reduced PINN
solver achieves a smaller error of $16.09\%$. In this higher-noise case,
the reduced PINN reconstruction is smoother and less sensitive to the
high-frequency oscillations in the noisy data, whereas the Tikhonov
reconstruction shows a more noticeable overshoot near the peak. Overall,
these results show that the Legendre--Tikhonov method provides a stable
reconstruction of the future profile from noisy initial observations, even
for the non-small final time $T=1$. The reduced PINN solver gives an
additional comparison after the same Legendre reduction.

{\bf Test 2: European butterfly spread payoff.}
In the second test, we use a European butterfly spread payoff, which is a
standard, financially meaningful compactly supported payoff. The final
time is set to $T=1.5$, and the terminal profile is defined by
\[
    \Phi(S)
    =
    \begin{cases}
    0, & 0\leq S<3,\\
    S-3, & 3\leq S<5,\\
    7-S, & 5\leq S<7,\\
    0, & S\geq 7.
    \end{cases}
\]
Equivalently, this is the payoff of a butterfly spread with strike prices
$K_1=3$, $K_2=5$, and $K_3=7$. Compared with the smooth profile in Test 1,
this payoff is less regular because it is piecewise linear and has kink
points at the strike prices. Hence, Test 2 evaluates the performance of
the Tikhonov and reduced PINN reconstructions in a more realistic setting
where the target future-time profile at maturity is generated from nonsmooth
financial data.

Figure~\ref{fig:test2} presents the numerical results for Test 2 with the
European butterfly spread payoff and final time $T=1.5$. In this experiment,
the observed initial data are contaminated by noise according to
\eqref{noise}, with two noise levels $5\%$ and $10\%$.
Figure~\ref{fig:test2a} and Figure~\ref{fig:test2c} show the exact and
noisy initial profiles, while Figure~\ref{fig:test2b} and
Figure~\ref{fig:test2d} compare the exact future-time profile $u(T,S)$
with the reconstructions obtained by the Tikhonov method and the reduced
PINN solver. Compared with Test 1, this example is more challenging
because the terminal butterfly payoff is only piecewise linear and
contains kink points at the strike prices. Moreover, the final time
$T=1.5$ is larger, which increases the instability of the forward-time
reconstruction problem.

\begin{figure}[ht]
    \centering
    \subfloat[\label{fig:test2a} Initial data, $\delta=5\%$]{
        \includegraphics[width=0.44\textwidth]{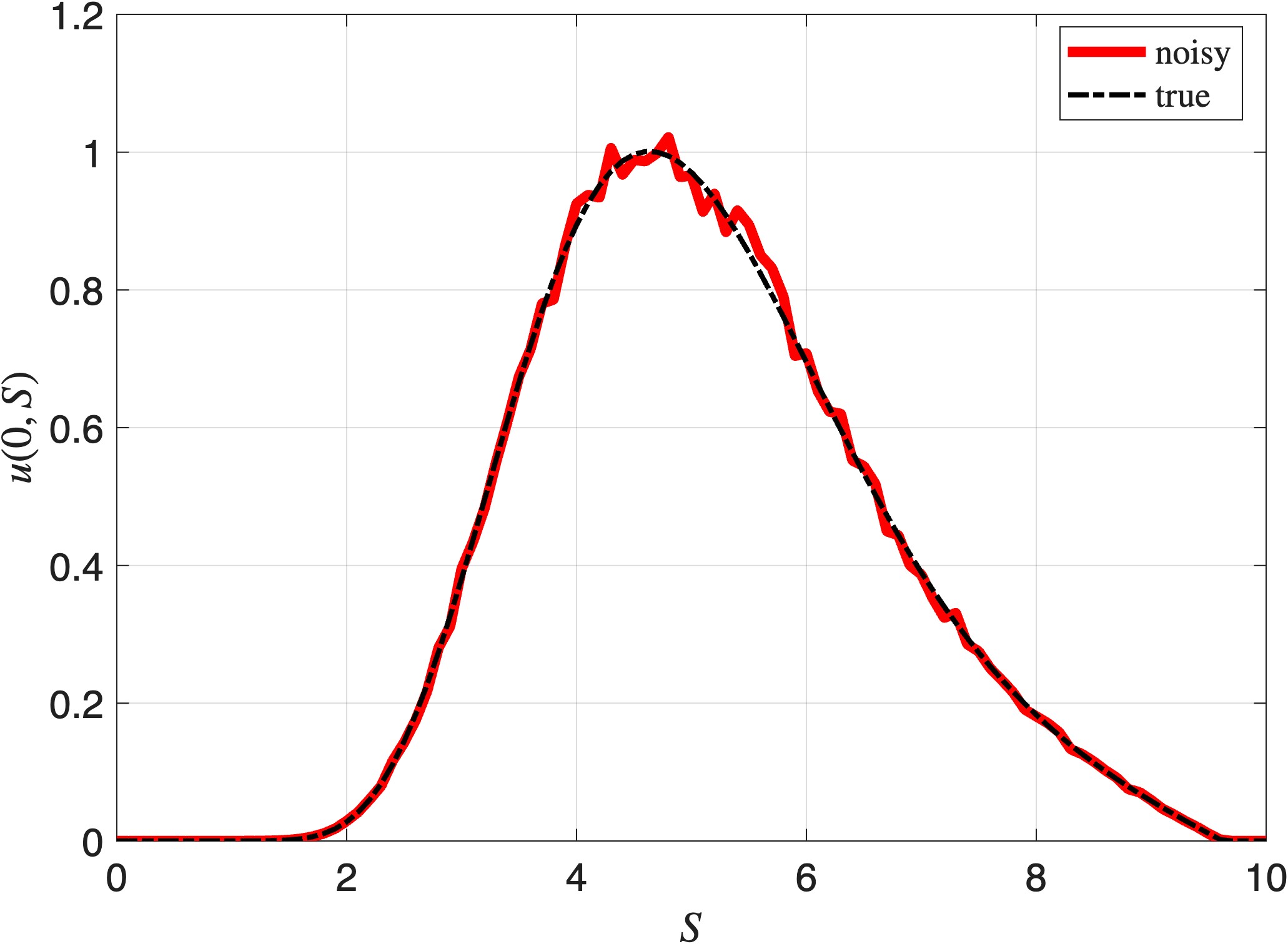}
    }
    \hfill
    \subfloat[\label{fig:test2b} Reconstruction of $u(T,S)$ from $u^0_\delta$, $\delta=5\%$]{
        \includegraphics[width=0.44\textwidth]{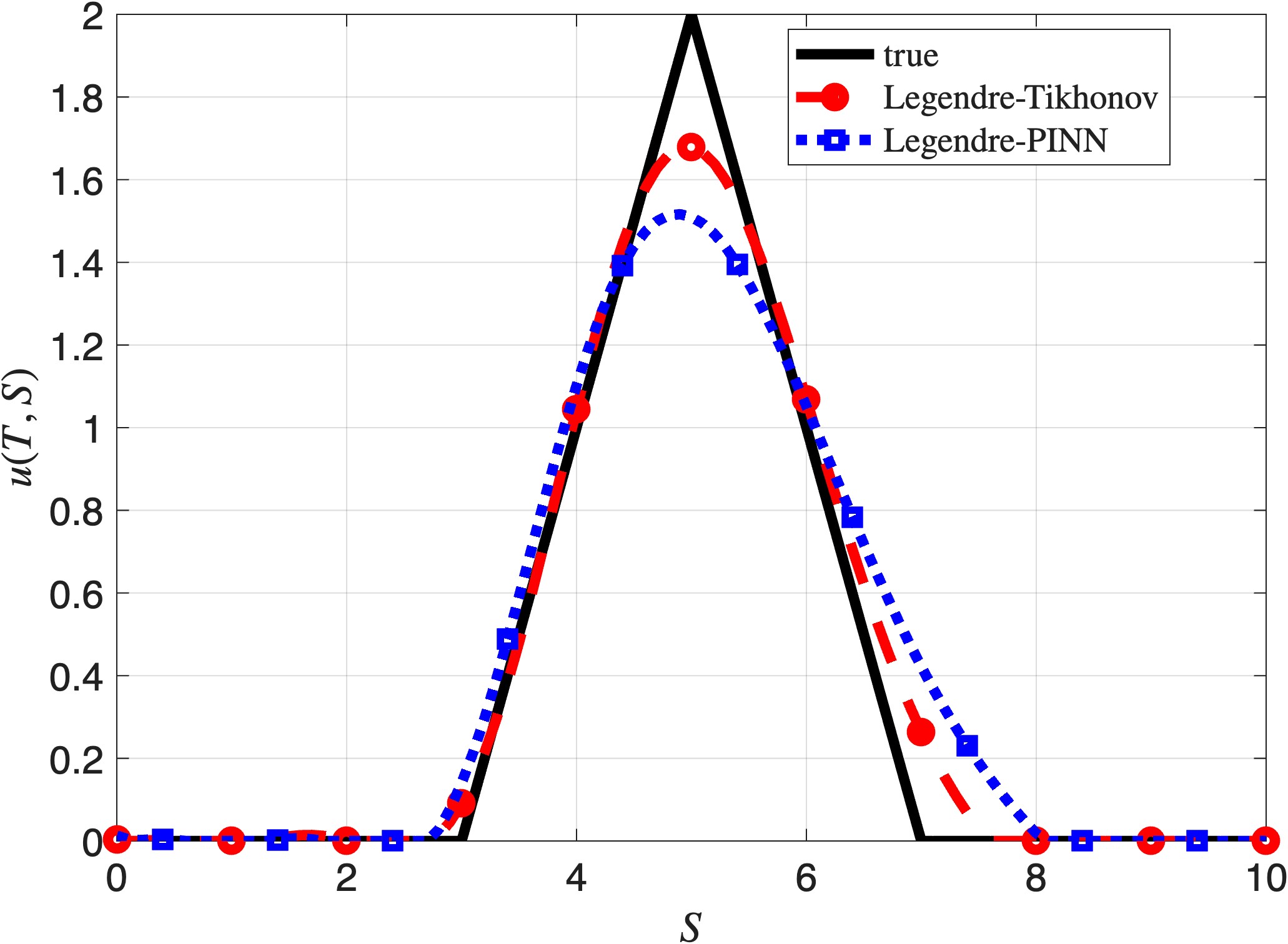}
    }
    
    \subfloat[\label{fig:test2c} Initial data, $\delta=10\%$]{
        \includegraphics[width=0.44\textwidth]{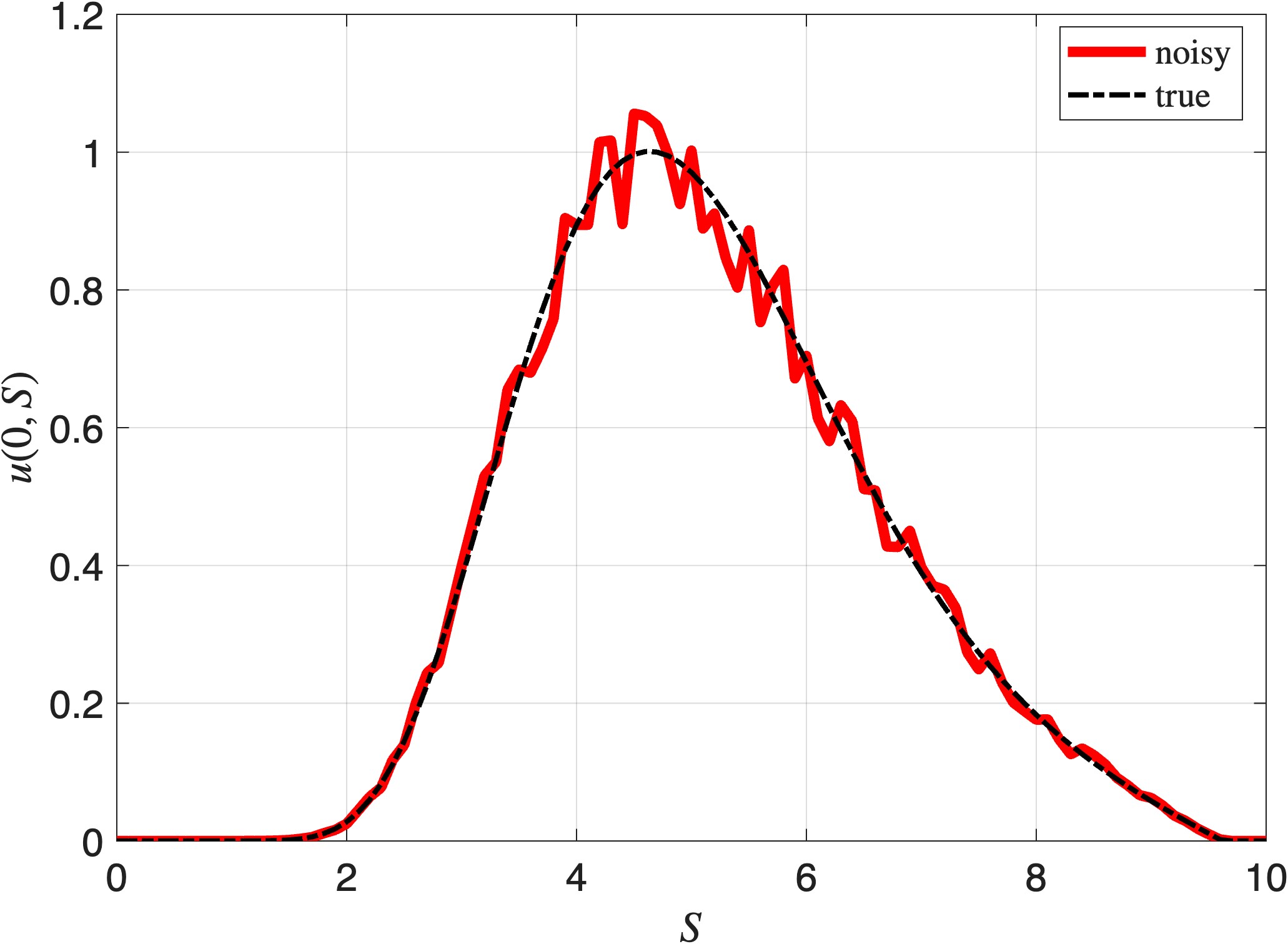}
    }
    \hfill
    \subfloat[\label{fig:test2d} Reconstruction of $u(T,S)$ from $u^0_\delta$, $\delta=10\%$]{
        \includegraphics[width=0.44\textwidth]{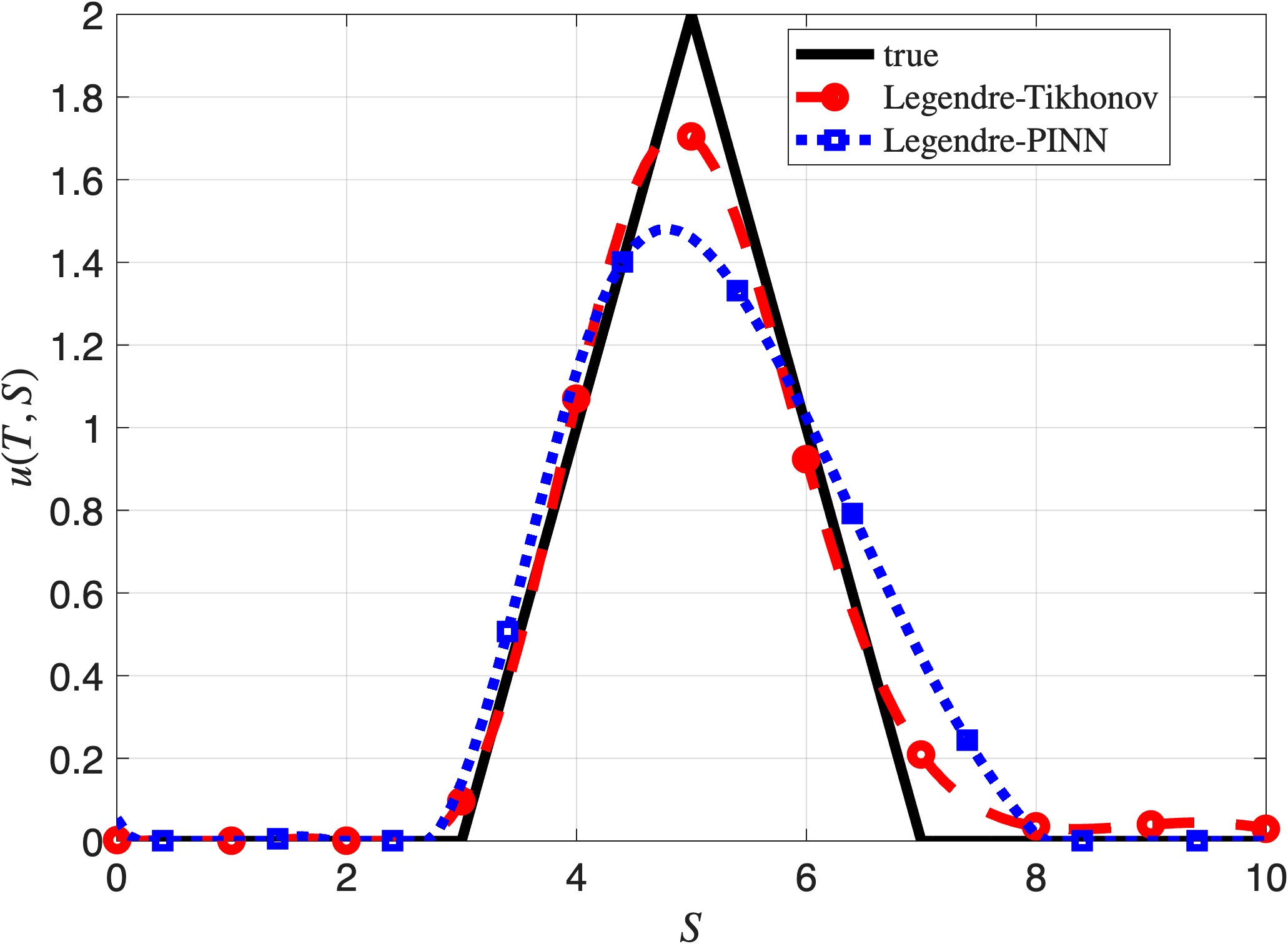}
    }
\caption{
Numerical results for Test 2 with final time $T=1.5$ and noisy initial data
generated according to \eqref{noise}. For $\delta=5\%$, the relative
$L^2$ errors are $11.05\%$ for the Tikhonov method and $20.53\%$ for the
reduced PINN solver. For $\delta=10\%$, the corresponding errors are
$10.22\%$ and $22.82\%$, respectively.
}
    \label{fig:test2}
\end{figure}

The results show that the Legendre--Tikhonov method and the reduced PINN solver
are able to recover the main structure of the future-time profile from noisy initial data. For the $5\%$ noise
case, the Tikhonov method achieves a relative $L^2$ error of $11.05\%$,
while the reduced PINN solver gives an error of $20.53\%$. For the
$10\%$ noise case, the corresponding errors are $10.22\%$ and $22.82\%$,
respectively. Thus, in this test, the Tikhonov method gives the more
accurate reconstruction in the relative $L^2$ sense for both noise
levels.

The Tikhonov reconstruction captures the peak height and the piecewise
linear shape of the future profile more closely. It also gives a sharper
approximation near the right edge of the support. The reduced PINN solver
still identifies the main support region and the peak location, but it
produces a smoother profile and tends to underestimate the maximum. This
smoothing effect is visible for both noise levels. Overall, Test 2
demonstrates that the proposed Legendre reduction remains effective for a
nonsmooth, financially meaningful payoff and for the final time $T=1.5$,
with the Tikhonov reconstruction performing particularly well in this
example.

We also observe that the Tikhonov error for the $10\%$ noise case is
slightly smaller than that for the $5\%$ noise case. This should not be
interpreted as a monotone improvement with respect to the noise level.
The noisy initial data are generated from random perturbations, and after
the Legendre truncation, the effective low-frequency component of the noise
may vary from one realization to another. Thus, for a single realization,
the reconstruction error need not increase monotonically with the nominal
noise level. The main observation is that the Legendre--Tikhonov method and the reduced PINN solver remain stable and
recover the dominant future-time structure in both noise regimes.

{\bf Test 3: European put payoff.}
In the third test, we consider a standard European put payoff. We set the
final time to $T=3$ and define the terminal profile by
\[
    \Phi(S)=(4-S)^+.
\]
Equivalently,
\[
    \Phi(S)
    =
    \begin{cases}
    4-S, & 0\leq S<4,\\
    0, & S\geq 4.
    \end{cases}
\]
This payoff has a kink at the strike price $K=4$. Compared with the
smooth profile in Test 1, it is less regular and therefore provides a
more realistic test for option-payoff reconstruction. The larger final
time $T=3$ also makes this example more challenging, since the
forward-time reconstruction becomes more unstable as the time horizon
increases.

Figure~\ref{fig:test4} presents the numerical results for Test 3 with the
European put payoff and final time $T=3$. The observed initial data are
perturbed according to \eqref{noise}, with noise levels $10\%$ and
$20\%$. Figure~\ref{fig:test4a} and Figure~\ref{fig:test4c} show the exact
and noisy initial profiles, while Figure~\ref{fig:test4b} and
Figure~\ref{fig:test4d} compare the exact future-time profile $u(T,S)$
with the reconstructions obtained by the Tikhonov method and the reduced
PINN solver. This test is challenging because the payoff has a kink at the
strike price and the final time $T=3$ is relatively large, which increases
the instability of the forward-time reconstruction problem.

\begin{figure}[ht]
    \centering
    \subfloat[\label{fig:test4a} Initial data, $\delta=10\%$]{
        \includegraphics[width=0.44\textwidth]{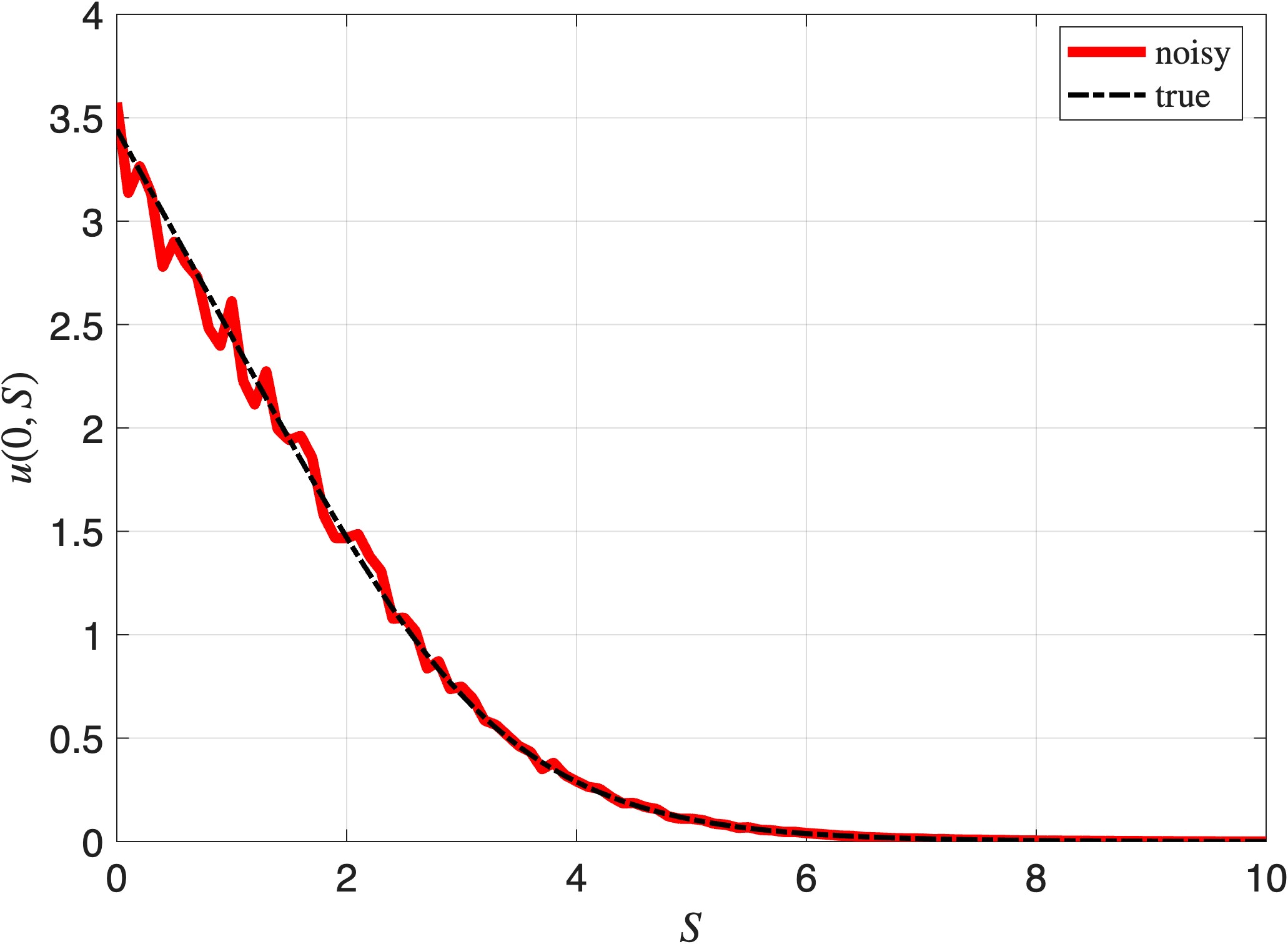}
    }
    \hfill
    \subfloat[\label{fig:test4b} Reconstruction of $u(T,S)$ from $u^0_\delta$, $\delta=10\%$]{
        \includegraphics[width=0.44\textwidth]{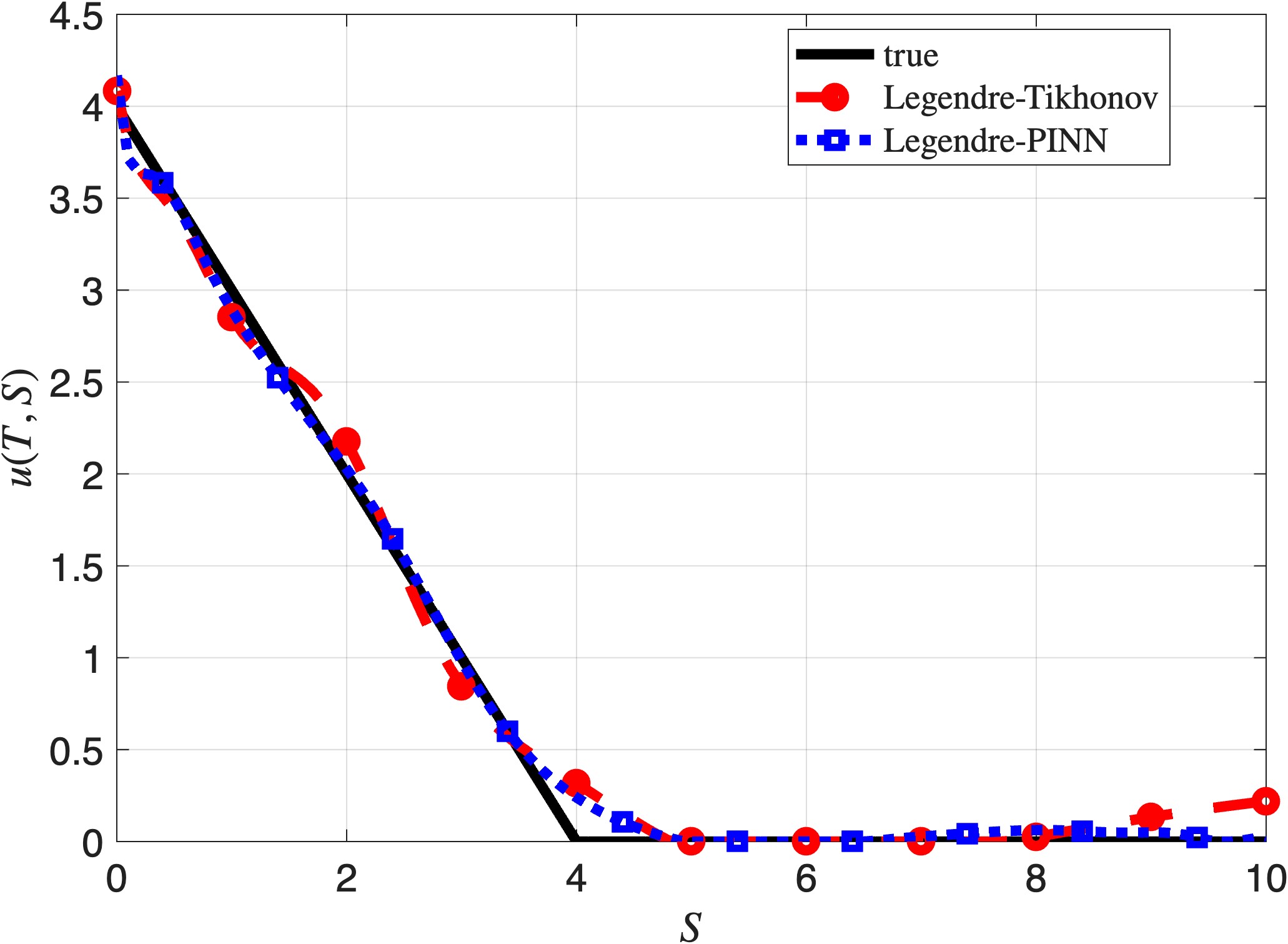}
    }
    
    \subfloat[\label{fig:test4c} Initial data, $\delta=20\%$]{
        \includegraphics[width=0.44\textwidth]{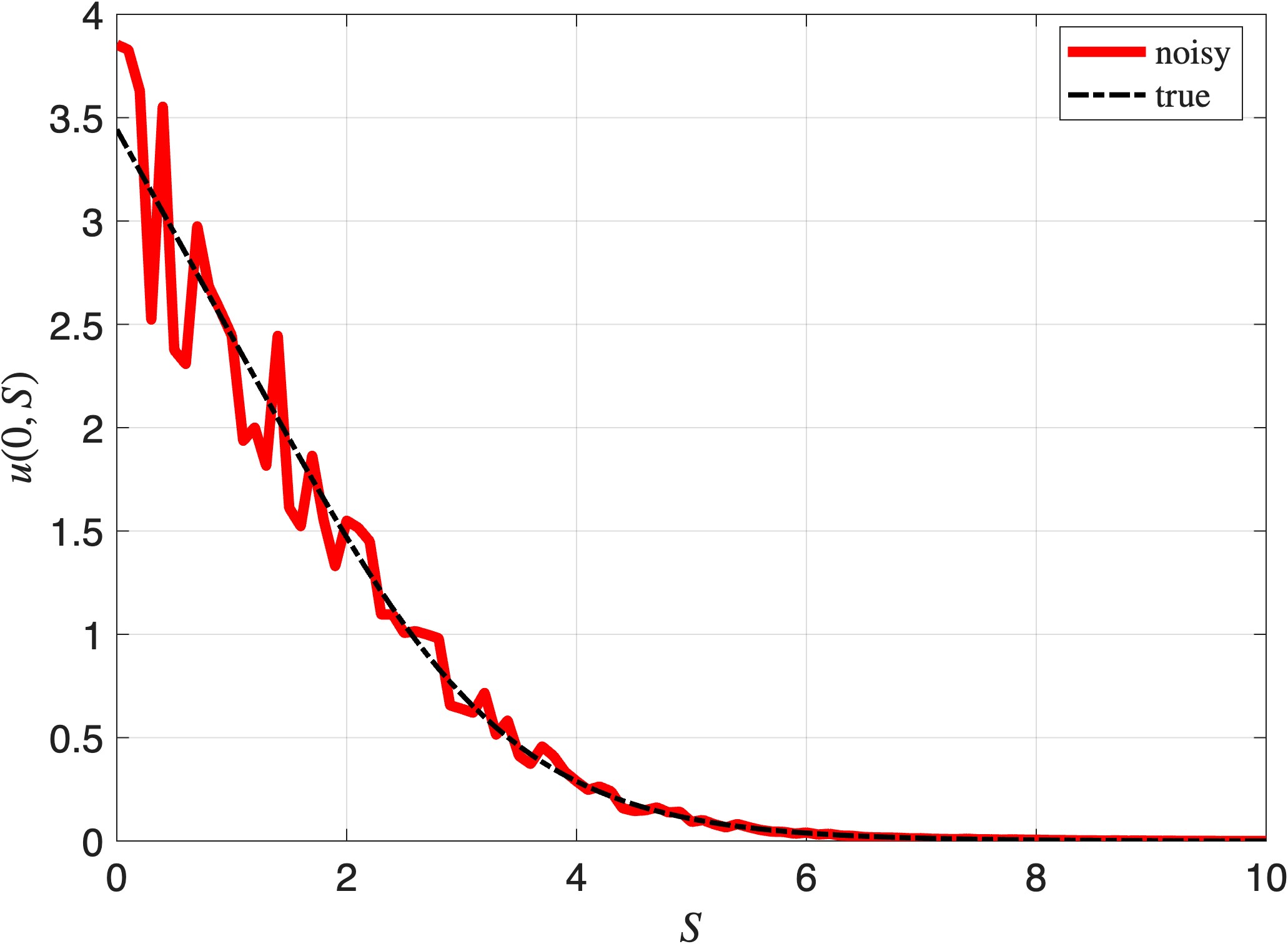}
    }
    \hfill
    \subfloat[\label{fig:test4d} Reconstruction of $u(T,S)$ from $u^0_\delta$, $\delta=20\%$]{
        \includegraphics[width=0.44\textwidth]{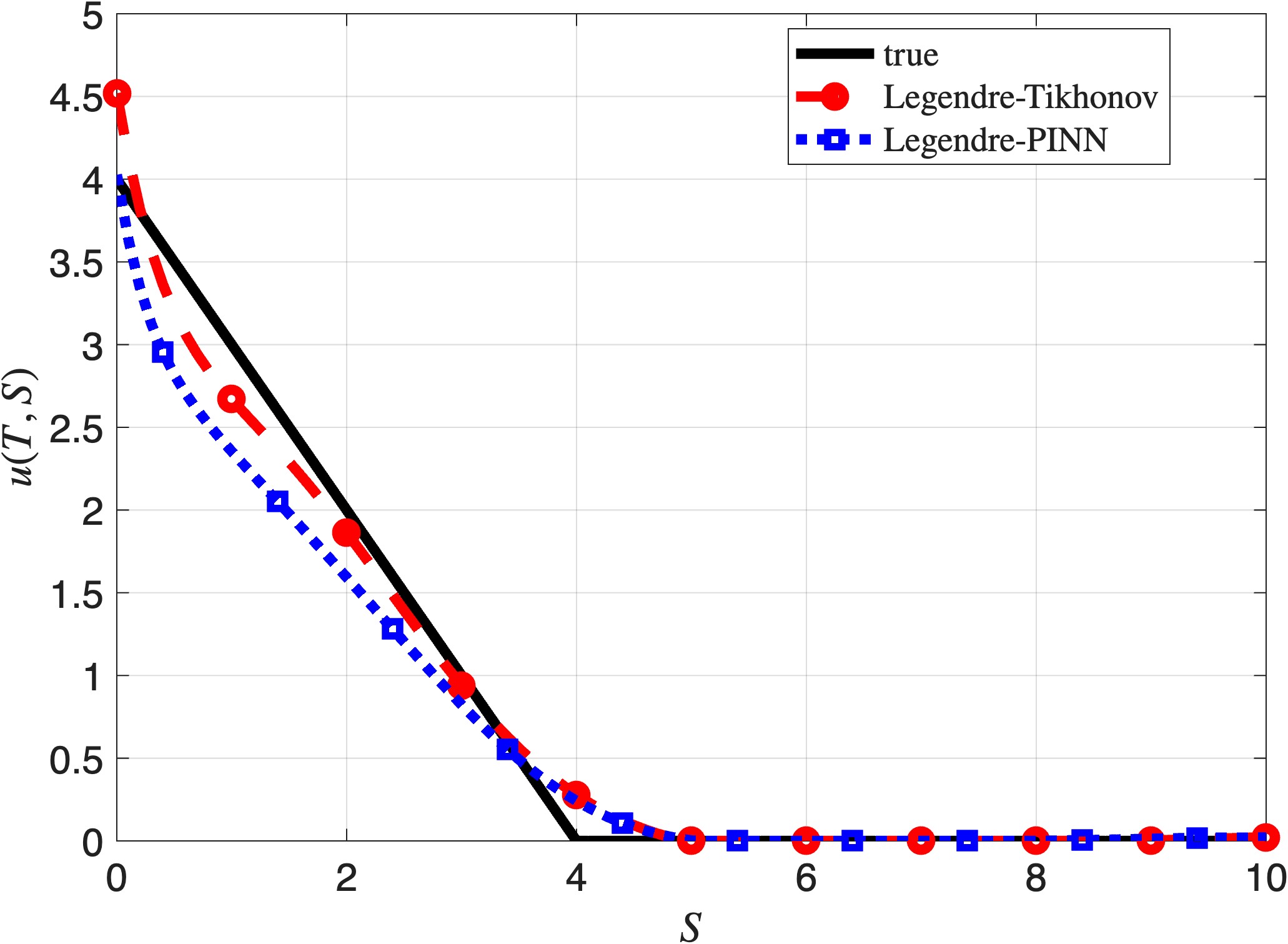}
    }
\caption{
Numerical results for Test 3 with final time $T=3$ and noisy initial data
generated according to \eqref{noise}. For $\delta=10\%$, the relative
$L^2$ errors are $7.83\%$ for the Tikhonov method and $4.36\%$ for the
reduced PINN solver. For $\delta=20\%$, the corresponding errors are
$9.37\%$ and $18.87\%$, respectively.
}
    \label{fig:test4}
\end{figure}

The results show that the Legendre--Tikhonov method and the reduced PINN solver
provide accurate reconstructions of the future-time profile. For the $10\%$ noise case, the relative $L^2$ errors
are $7.83\%$ for the Tikhonov method and $4.36\%$ for the reduced PINN
solver. For the $20\%$ noise case, the corresponding errors are $9.37\%$
and $18.87\%$, respectively. Thus, in this test, the reduced PINN solver
gives the smaller error for the $10\%$ noise case, while the
Legendre--Tikhonov method is more accurate for the $20\%$ noise case.
Both reconstructions capture the decreasing structure of the put payoff
and the location of the kink near the strike price. Although small
oscillations appear near the far-field region, the overall reconstruction
quality remains stable despite the larger final time $T=3$ and the
ill-posedness of the problem.

\FloatBarrier

\begin{table}[H]
\centering
\caption{Summary of relative $L^2$ reconstruction errors for the numerical tests.}
\label{tab:numerical-errors}
\small\begin{tabular}{c p{2.2in} c c c c}
\hline
Test & Terminal profile & $T$ & Noise level & Tikhonov error & Reduced PINN error \\
\hline
1 & Smooth compactly supported profile & $1$ & $10\%$ & $14.56\%$ & $15.40\%$ \\
1 & Smooth compactly supported profile & $1$ & $35\%$ & $23.27\%$ & $16.09\%$ \\
2 & European butterfly spread & $1.5$ & $5\%$ & $11.05\%$ & $20.53\%$ \\
2 & European butterfly spread & $1.5$ & $10\%$ & $10.22\%$ & $22.82\%$ \\
3 & European put payoff & $3$ & $10\%$ & $7.83\%$ & $4.36\%$ \\
3 & European put payoff & $3$ & $20\%$ & $9.37\%$ & $18.87\%$ \\
\hline
\end{tabular}
\end{table}

\FloatBarrier

\begin{Remark}
The numerical results show that the reduced PINN solver can provide
reasonable reconstructions after Legendre reduction, but it is used here
only as a secondary computational comparison. The main method of the paper
is the Legendre--Tikhonov reconstruction, for which the stability and
fixed-$N$ convergence theory was proved in Theorem~\ref{thm:tikhonov-fixed-N}.
The Reduced PINN errors in Table~\ref{tab:numerical-errors} should
therefore be interpreted as empirical evidence for the usefulness of a
neural-network parametrization of the reduced coefficient system, not as a
separate regularization theorem.
\end{Remark}

\subsection{Choice of the parameters $\alpha$ and $N$}\label{alphaN}

We now discuss the selection of the regularization parameter $\alpha$ and
the truncation number $N$ in the dimension-reduced Legendre--Tikhonov
method. The procedure is illustrated using Test 1 with final time $T=1$
and $10\%$ noise in the initial data. The curves in
Figure~\ref{fig:choice-alpha-N} are computed using only the noisy initial
profile $u^0_\delta$ and the reduced coefficient system. The exact future
profile $u_{\rm true}(T,S)$ is used only to report the reconstruction
error in the synthetic experiments, and is not used in the selection of
either $\alpha$ or $N$.

For a fixed truncation number $N$, the parameter $\alpha$ is chosen with
the guidance of the L-curve criterion. For each candidate value of
$\alpha$, we solve the dimension-reduced Legendre--Tikhonov problem and
compute the residual quantity
\[
    R(\alpha)
    =
    \left(
    \int_0^T
    \left|
        \mathbf u_\alpha'(t)-C(t)\mathbf u_\alpha(t)
    \right|^2\,dt
    +
    \left|
        \mathbf u_\alpha(0)-\mathbf u^0_\delta
    \right|^2
    \right)^{1/2}
\]
and the regularization norm
\[
    Q(\alpha)
    =
    \|\mathbf u_\alpha\|_{H^2(0,T;\mathbb R^{N+1})}.
\]
The L-curve is obtained by plotting $R(\alpha)$ against $Q(\alpha)$ for
the candidate values of $\alpha$. This curve shows the trade-off between
fitting the reduced system and the projected noisy initial data, and
controlling the regularity of the coefficient vector. In
Figure~\ref{fig:lcurve-case1-noise10}, the L-curve provides a stable
range for choosing $\alpha$. In the reported reconstruction for this
example, we choose $\alpha=3.2\times 10^{-5}$, which gives sufficient
regularization to suppress oscillations from the noisy initial data.

The truncation number $N$ is selected by examining the residual quantity
for different values of $N$. With the regularization parameter fixed, we
solve the dimension-reduced Legendre--Tikhonov problem for each candidate
$N$ and compute
\[
    R_N
    =
    \left(
    \int_0^T
    \left|
        \mathbf u_N'(t)-C_N(t)\mathbf u_N(t)
    \right|^2\,dt
    +
    \left|
        \mathbf u_N(0)-\mathbf u_{\delta,N}^0
    \right|^2
    \right)^{1/2}.
\]
Here $C_N(t)$ is the coefficient matrix obtained by retaining the first
$N+1$ Legendre modes, and $\mathbf u_{\delta,N}^0$ is the projection of
the noisy initial data onto the same reduced space. We choose $N$ near
the minimum of the curve $N\mapsto R_N$. In
Figure~\ref{fig:residual-vs-N-case1-noise10}, the residual reaches its
smallest value near $N=15$, and this value is used for the corresponding
reconstruction.

\begin{figure}[ht]
    \centering

    \subfloat[L-curve for choosing $\alpha$.]
    {
        \includegraphics[width=0.44\textwidth]{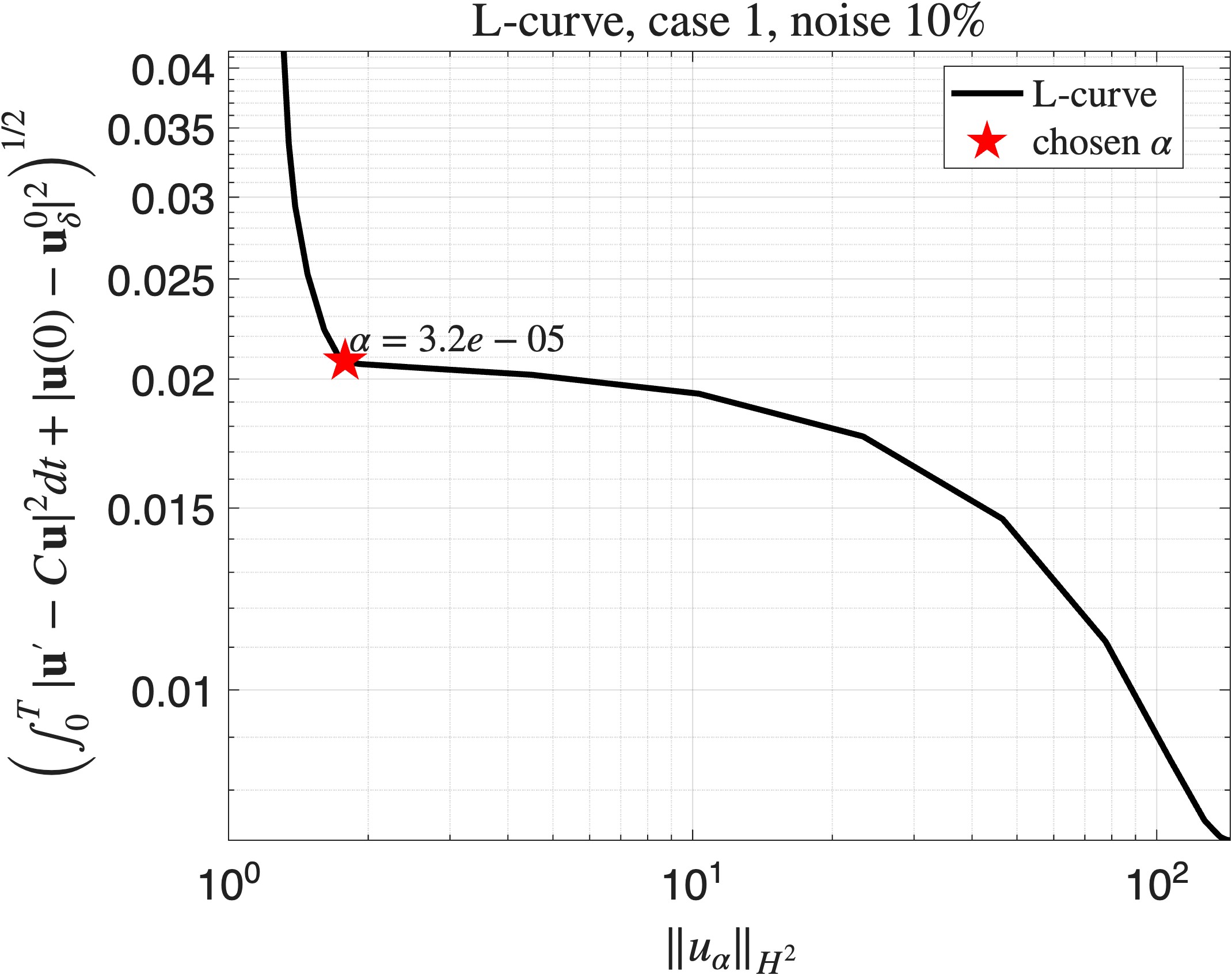}
        \label{fig:lcurve-case1-noise10}
    }
    \hfill
    \subfloat[Residual quantity versus $N$.]
    {
        \includegraphics[width=0.44\textwidth]{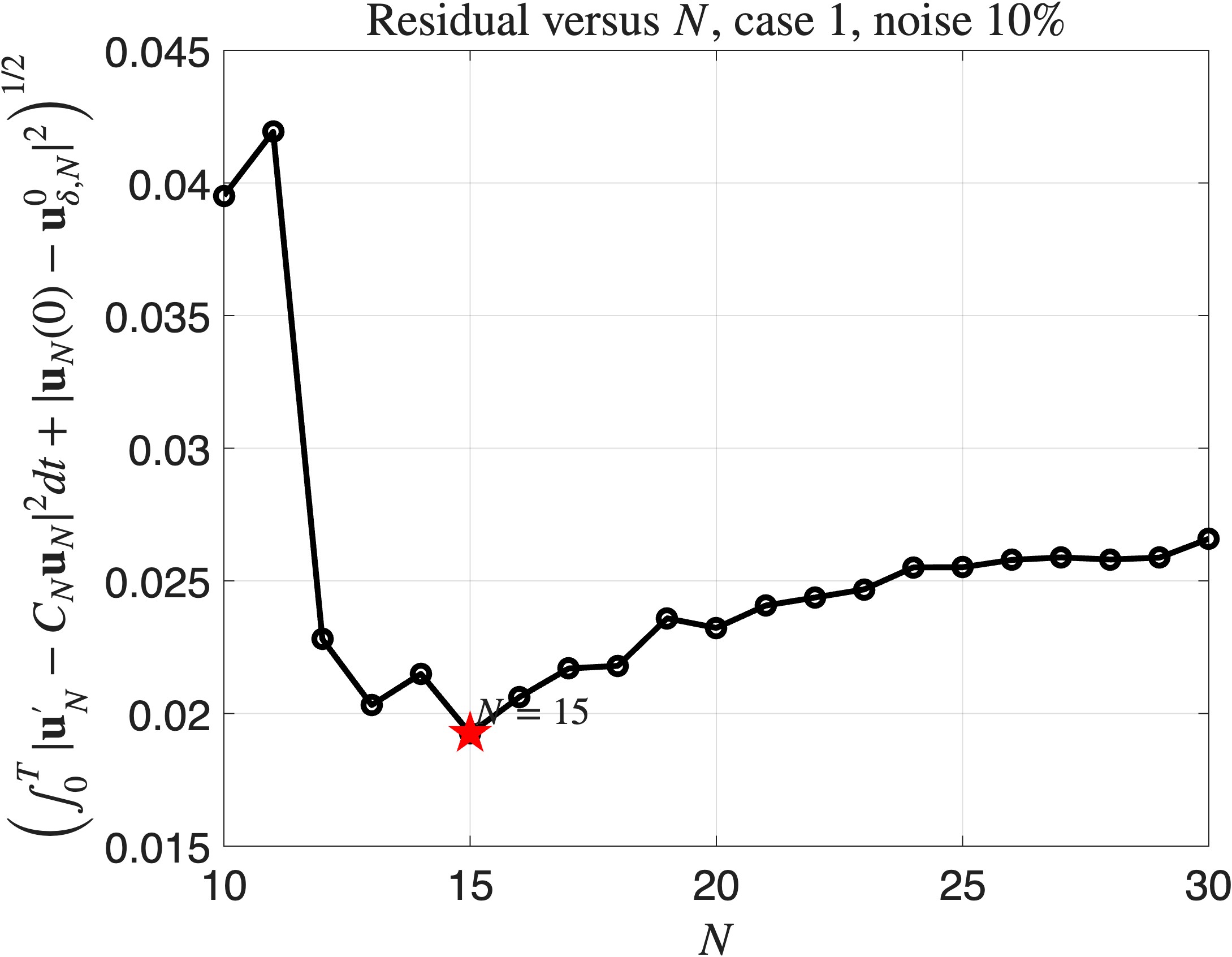}
        \label{fig:residual-vs-N-case1-noise10}
    }

    \caption{
    Parameter selection for Test 1 with $10\%$ noise. The L-curve in
    Figure~\ref{fig:lcurve-case1-noise10} is used to guide the choice of
    the regularization parameter $\alpha$. The residual curve in
    Figure~\ref{fig:residual-vs-N-case1-noise10} is used to choose the
    truncation number $N$. In this example, we choose
    $\alpha=3.2\times 10^{-5}$ and $N=15$.
    }
    \label{fig:choice-alpha-N}
\end{figure}

\subsection{Effect of an incorrect choice of $N$}

We next illustrate the role of the truncation number $N$ in the
Legendre-reduction method. As discussed in Section~2, the forward-time
Black--Scholes problem is unstable because high-frequency components of
the initial data may be strongly amplified as time evolves. The Legendre
projection helps control this instability by acting as a spectral cutoff:
only the first $N+1$ modes are retained, while higher-order oscillatory
components are removed from the reduced coefficient system.

The choice of $N$ is therefore part of the regularization mechanism. If
$N$ is chosen too small, important low-frequency information may be lost
and the reconstruction may be overly smoothed. On the other hand, if $N$
is chosen too large, the truncated expansion begins to retain modes that
are dominated by noise. These high-order modes can reintroduce the
unstable oscillations described in Section~2 and degrade the future-time
reconstruction.

This effect is shown in Figure~\ref{fig:wrong-N-case1-noise10}. We repeat
Test~1 with $10\%$ noise in the initial data and intentionally choose two
inappropriate truncation numbers. The choice $N=6$ is too small and
under-resolves the future-time profile. The excessive choice $N=120$ does
not merely improve the resolution of the Legendre approximation; instead,
it retains noisy high-order components, producing visible oscillations
and an overestimated peak in the recovered future-time profile.

\begin{figure}[ht]
    \centering
    \subfloat[$N=6$]{
        \includegraphics[width= .44\textwidth, height=0.25\textheight]{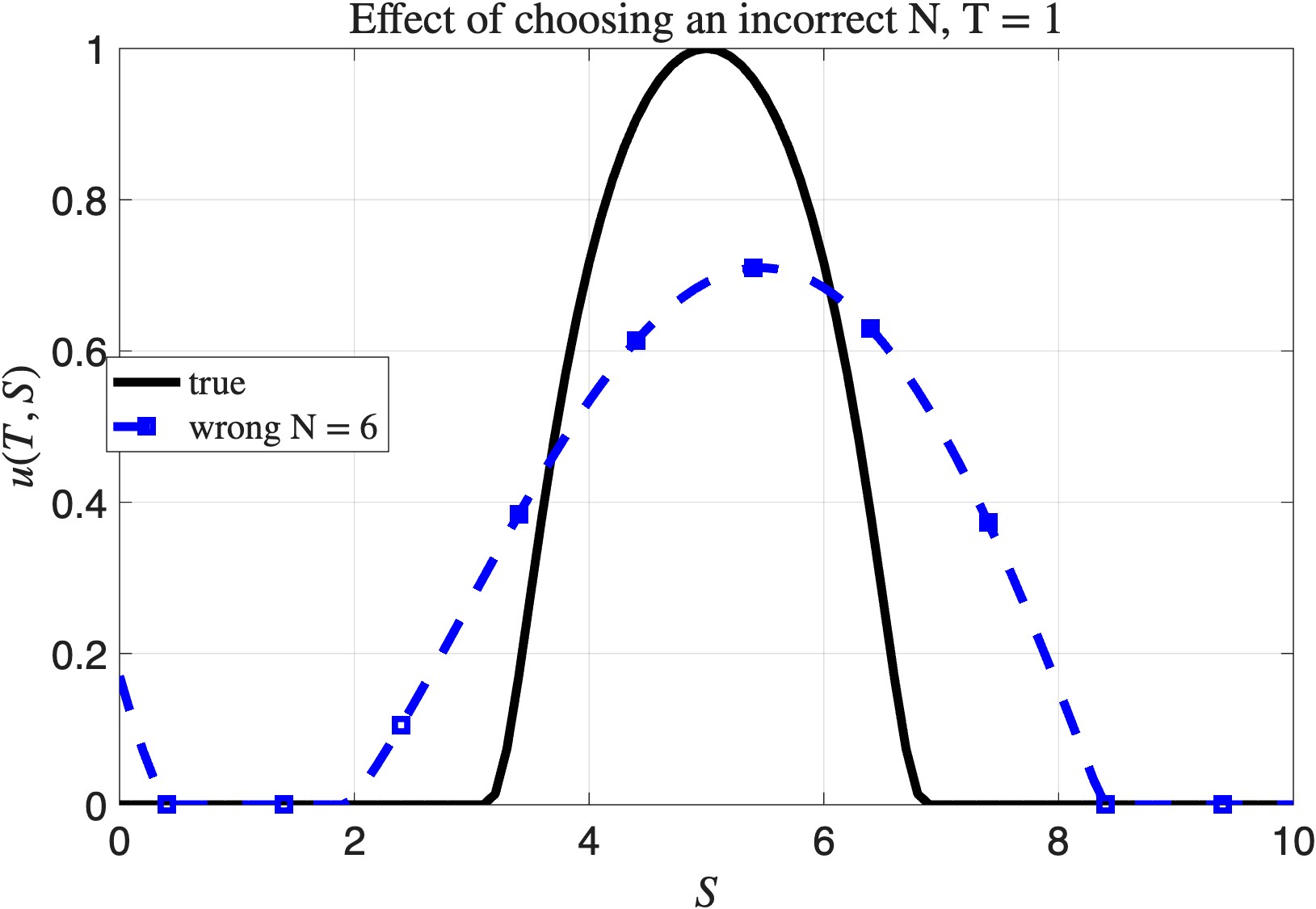}
    }
    \quad
    \subfloat[$N=120$]{
        \includegraphics[width= .44\textwidth, height=0.25\textheight]{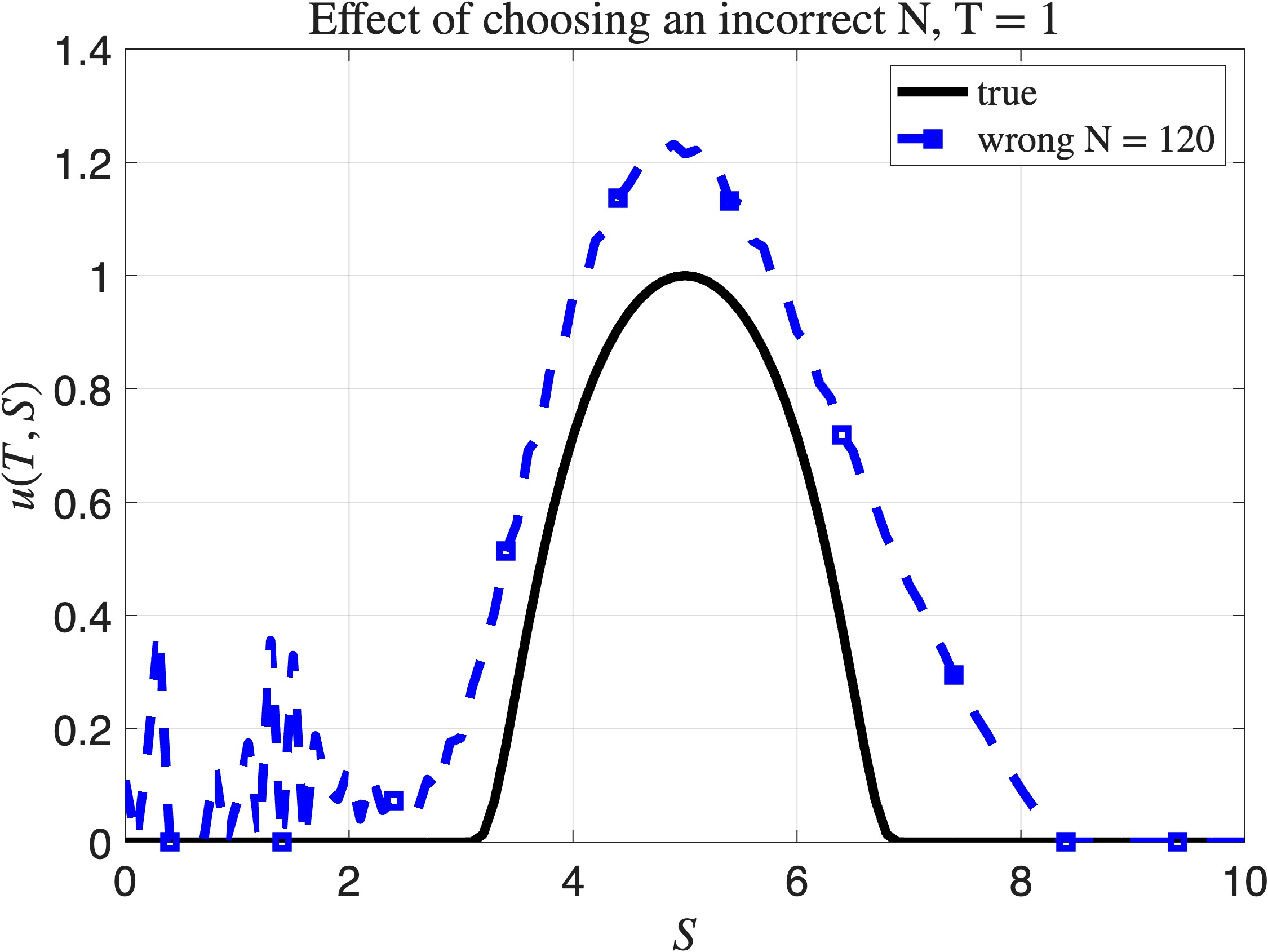}
    }
    \caption{
    Effect of incorrectly chosen truncation numbers for Test~1 with
    $10\%$ noise. The choice $N=6$ is too small and under-resolves the
    future-time profile, while the excessive choice $N=120$ retains
    unstable high-order Legendre modes and produces spurious oscillations.
    This illustrates the importance of choosing a balanced truncation
    number.
    }
    \label{fig:wrong-N-case1-noise10}
\end{figure}

This example confirms that the truncation number $N$ is not only a
discretization parameter but also a regularization parameter. A suitable
choice of $N$ balances approximation accuracy and stability: it preserves
the dominant low-order structure of the option-price profile while
discarding high-order modes that mainly represent noise. This observation
supports the central idea of the proposed dimension-reduced
Legendre--Tikhonov method, namely that the spectral cutoff induced by the
Legendre reduction helps stabilize the ill-posed forward-time
reconstruction.

\subsection{Comparison with the conventional quasi-reversibility method}

We next compare the proposed Legendre-reduction approach with the
conventional least-squares formulation, also known as the
quasi-reversibility method. In the conventional approach, the
Black--Scholes equation is solved directly in the physical variables
$(t,S)$. More precisely, we approximate $u(t,S)$ on a space-time grid and
determine the discrete solution by minimizing a Tikhonov-type
least-squares functional consisting of the Black--Scholes residual, the
mismatch with the initial data, and an $H^2$ regularization term. In
continuous notation, this corresponds to minimizing
\begin{multline}
    J_\alpha^{\rm phys}(u)
    =
    \int_0^T\int_0^{S_{\rm max}}
    \left|
        u_t
        +\frac12\sigma^2(t,S)S^2u_{SS}
        +rS u_S
        -ru
    \right|^2\,dSdt  \\
    +
    \|u(0,\cdot)-u_0\|_{L^2(0,S_{\rm max})}^2
    +
    \alpha\|u\|_{H^2((0,T)\times(0,S_{\rm max}))}^2 .
\end{multline}
The resulting linear least-squares system is solved by the LSQR algorithm.
This physical-space quasi-reversibility method is used as a baseline for
comparison with the proposed method. In contrast, our approach first
projects the solution onto a finite-dimensional Legendre basis in the
price variable $S$ and then solves the resulting reduced system of
ordinary differential equations in time.

To isolate the effect of the final time on the reconstruction quality, we
use the same smooth terminal profile as in Test 1 and repeat the
experiment in the noise-free setting for three final times:
$T=0.1$, $T=0.2$, and $T=0.3$. Thus, in
Figure~\ref{fig:compare-qrm}, the true profile and all model parameters
are fixed, while only the value of $T$ is varied. Since no noise is added
to the initial data, the comparison reflects the numerical stability of
the two formulations rather than the effect of data perturbation. Figure~\ref{fig:compare-qrm} shows the reconstructions for
$T=0.1$, $T=0.2$, and $T=0.3$.

\begin{figure}[ht]
    \centering

    \subfloat[Noise-free reconstruction, $T=0.1$.]
    {
        \includegraphics[width=0.3\textwidth]{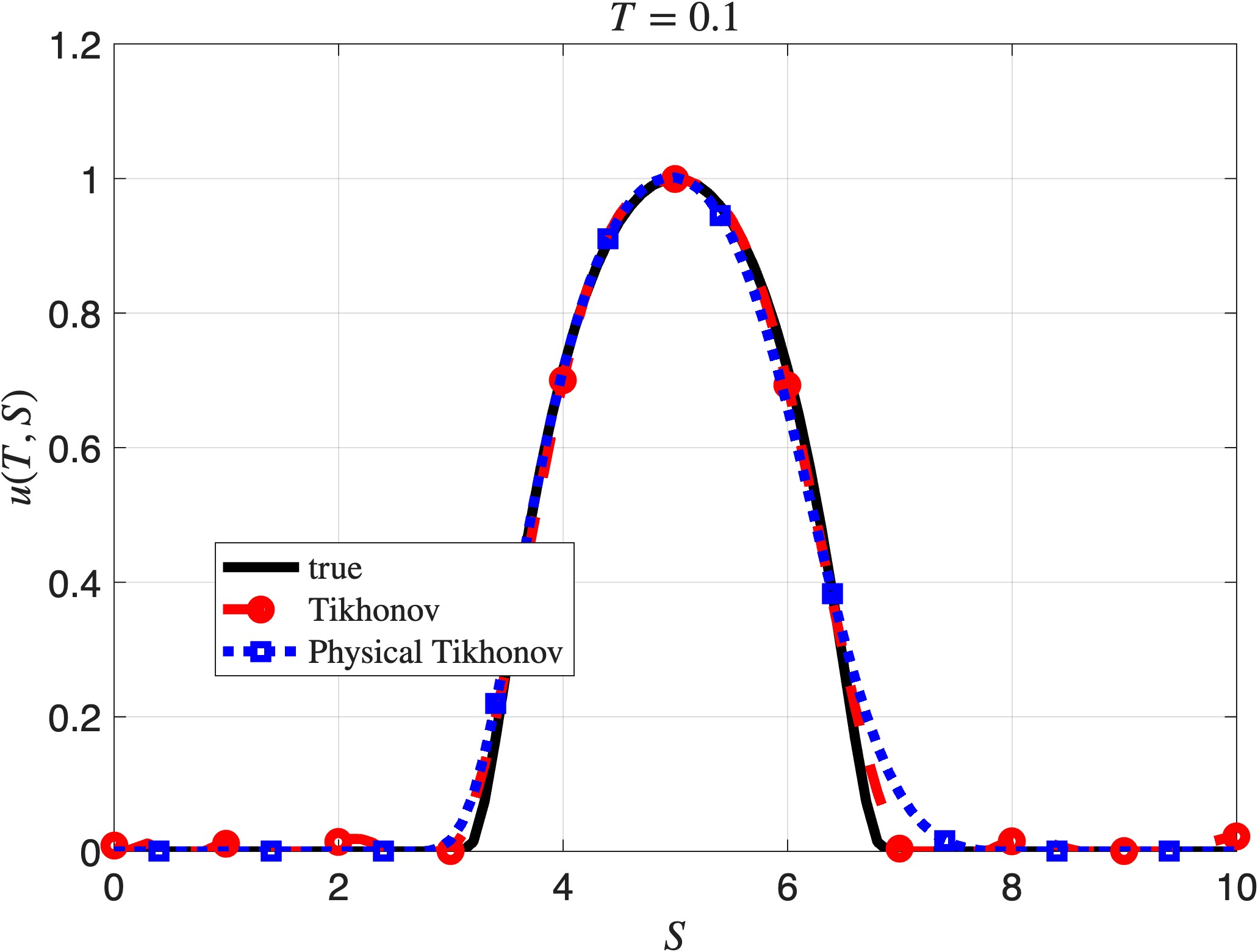}
        \label{fig:compare-physical-T01}
    }
    \hfill
    \subfloat[Noise-free reconstruction, $T=0.2$.]
    {
        \includegraphics[width=0.3\textwidth]{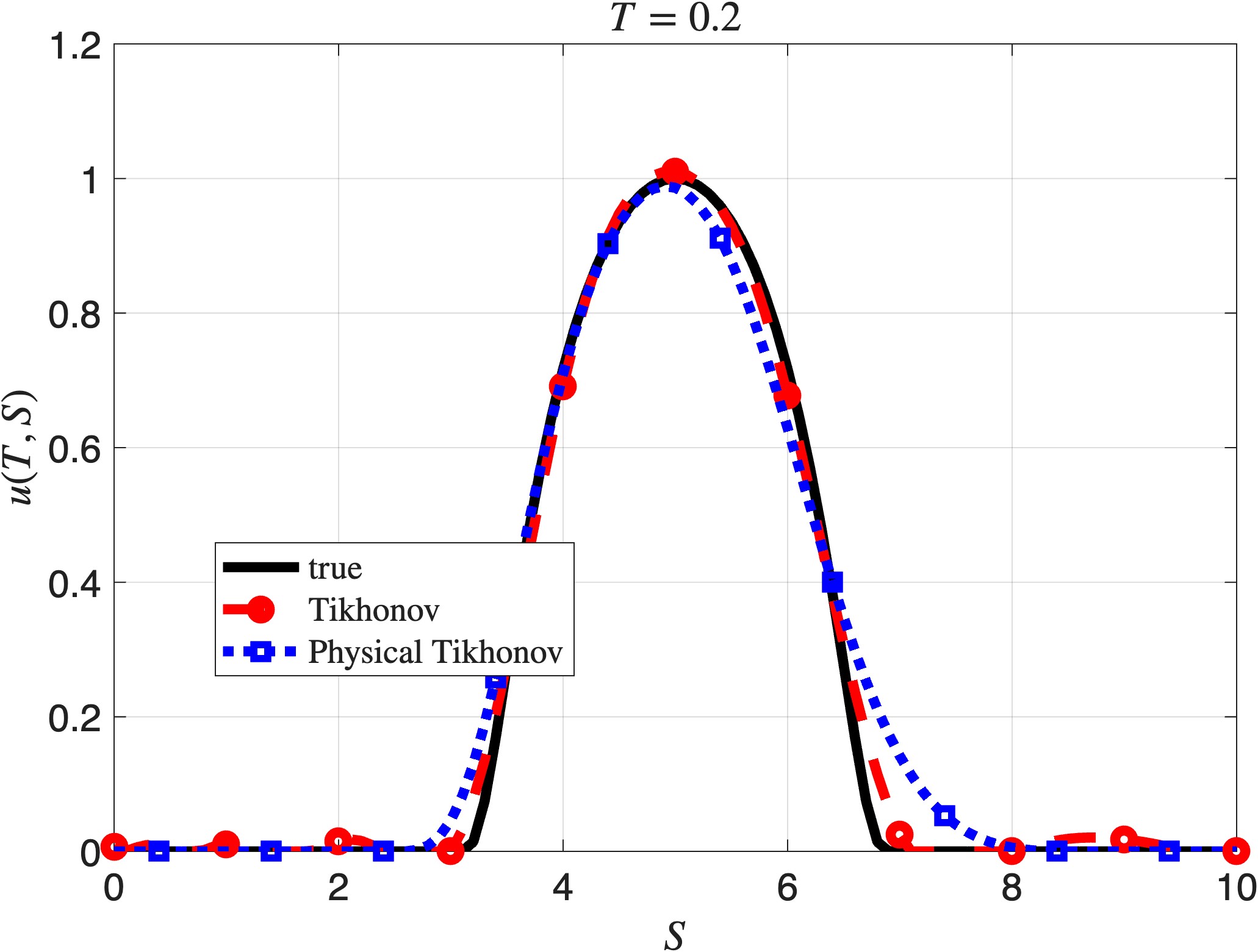}
        \label{fig:compare-physical-T02}
    }
    \hfill
    \subfloat[Noise-free reconstruction, $T=0.3$.]
    {
        \includegraphics[width=0.3\textwidth]{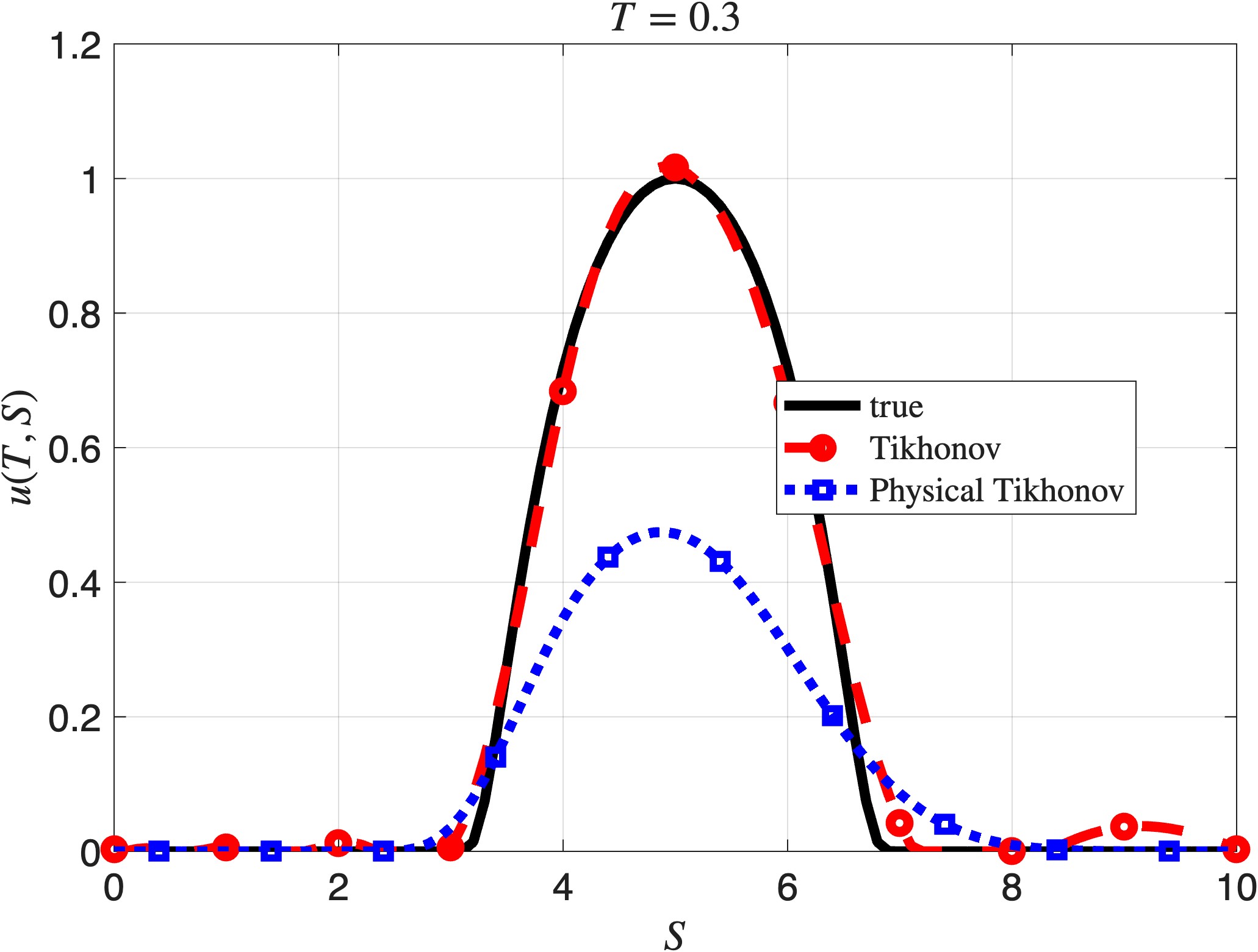}
        \label{fig:compare-physical-T03}
    }

    \caption{
    Comparison between the proposed Legendre-reduction method and the
    conventional least-squares quasi-reversibility method in the
    noise-free case. For $T=0.1$ and $T=0.2$, the Legendre--Tikhonov method and the reduced PINN solver provide
    accurate reconstructions. When $T=0.3$, the conventional method starts
    to lose accuracy and significantly underestimates the peak, whereas
    the Legendre-reduction method remains close to the exact future-time
    profile.
    }
    \label{fig:compare-qrm}
\end{figure}

Figure~\ref{fig:compare-qrm} shows the reconstructions for
$T=0.1$, $T=0.2$, and $T=0.3$. For the shorter final times
$T=0.1$ and $T=0.2$, the conventional quasi-reversibility method
still gives reasonable approximations of the exact future profile.
This suggests that the conventional method can be effective when the
time interval is sufficiently short and the instability has not yet
become dominant. However, as the final time increases to $T=0.3$, the
performance of the conventional method deteriorates: the reconstructed
profile becomes overly smoothed and significantly underestimates the
amplitude of the true solution. In contrast, the Legendre--Tikhonov
method and the reduced PINN solver remain closer to the exact profile
and continue to capture both the location and the height of the main
peak. This comparison indicates that the Legendre reduction provides a
stabilizing effect for longer time intervals, where the conventional
quasi-reversibility method becomes less reliable.

These results demonstrate the benefit of the Legendre reduction. Although
the original forward-time Black--Scholes problem is ill-posed, projecting
the solution onto a finite number of Legendre modes removes the most
unstable high-frequency components and converts the problem into a
finite-dimensional system in time. This reduced formulation provides a
more stable numerical framework than applying quasi-reversibility directly
to the original partial differential equation. The comparison suggests
that the Legendre-reduction step is especially important when the final
time is not very small.

\begin{Remark}[Use of noise-free data in the comparison]
The comparison in Figure~\ref{fig:compare-qrm} is carried out
with noise-free initial data. This choice is intentional. The purpose of
this comparison is to isolate the effect of the Legendre reduction from
the effect of measurement noise. The conventional physical-space
Tikhonov/quasi-reversibility formulation is applied directly to the
Black--Scholes equation in the variables $(t,S)$, and therefore it remains
highly sensitive to the unstable high-frequency components discussed in
Section~\ref{ill-posedness}. In our experiments, when noisy initial data
are used, the physical-space reconstruction becomes dominated by noise
amplification even for very small final times such as $T=0.1$. Such
results do not provide a useful visual comparison of the two formulations,
because the instability of the physical-space method overwhelms the
underlying reconstructed profile.

For this reason, the noise-free setting is used in this comparison as a
controlled diagnostic test. It allows us to examine whether the two
methods can recover the future-time profile when the only source of
difficulty is the forward-time instability of the equation itself, rather
than additional measurement noise. The noisy-data experiments in the
preceding subsections are used to demonstrate the robustness of the
proposed Legendre--Tikhonov reconstruction under perturbed initial data.
Together, these tests show that the Legendre reduction improves stability
both by removing high-order oscillatory modes and by providing a more
favorable reduced system for subsequent regularization.
\end{Remark}

\section{Concluding remarks}\label{sec6}

In this paper, we studied a forward-time formulation of the Black--Scholes
equation with state-dependent volatility. Unlike the classical
terminal-value pricing problem, the present problem prescribes the current
option-price profile and seeks to recover the option price profile at maturity $T$. This formulation is ill-posed because it evolves the parabolic
operator in the unstable direction, and high-frequency perturbations in
the initial data may be strongly amplified.

To address this difficulty, we introduced a price-dimensional reduction
based on shifted Legendre polynomials. By projecting the solution in the
asset-price variable, the original Black--Scholes equation is transformed
into a finite-dimensional system of ordinary differential equations in
time for the Legendre coefficient vector. This reduction acts as a
spectral cutoff and removes the most unstable high-frequency modes. It
also relaxes the difficulty caused by the degeneracy of the coefficient
$\frac12\sigma^2(t,S)S^2$ at $S=0$, since the principal differential
operator in the reduced system is the first derivative with respect to
time.

The main reconstruction method in this paper is the dimension-reduced
Legendre--Tikhonov method. The method uses a Tikhonov functional involving
the reduced differential residual, the projected initial data, and an
$H^2$ penalty in time. We proved existence, uniqueness, data stability,
and convergence for each fixed truncation level $N$. This provides a
rigorous regularization justification for the proposed reduced
reconstruction framework.

We also tested a reduced PINN solver as a secondary numerical approach.
The PINN is applied only to the reduced Legendre coefficient system, so it
approximates $N+1$ scalar functions of the time variable rather than the
full two-variable function $u(t,S)$. The PINN results suggest that
neural-network parametrizations can be useful after Legendre reduction,
but the rigorous stability and convergence analysis in this paper is
carried out for the Legendre--Tikhonov method.

The numerical experiments demonstrate that the Legendre--Tikhonov method
can recover the main qualitative features of the terminal option--price profile from noisy initial data. The tests include a smooth compactly
supported profile, a European butterfly spread, and a European put payoff.
These examples cover smooth and nonsmooth terminal profiles, as well as
compactly supported and non-compact payoff structures. The error summary
in Table~\ref{tab:numerical-errors} shows that the reduced PINN solver is
competitive in some cases, but it should be viewed as a supplementary
computational benchmark rather than as the main regularization method.

We also compared the proposed Legendre-reduction approach with the
conventional physical-space quasi-reversibility method. In the noise-free
comparison, the conventional method performs well for very small final
times, but starts to lose accuracy as the final time increases. In
contrast, the Legendre-reduction method remains stable for the tested
examples. This comparison supports the main idea of the paper: reducing
the price dimension before applying regularization is beneficial for the
forward-time Black--Scholes problem.

\section*{Declarations}

\medskip
\noindent\textbf{Declaration of competing interest.}
The authors declare that they have no known competing financial interests or personal relationships that could have appeared to influence the work reported in this paper.

\medskip
\noindent\textbf{Data availability.}
The data that support the findings of this study are available from the corresponding author upon reasonable request.


\end{document}